\documentclass[11pt]{article}
\usepackage{doi}

\usepackage[T1]{fontenc}
\usepackage[utf8]{inputenc}
\usepackage[margin=1in]{geometry}
\usepackage{microtype}
\usepackage{authblk}

\usepackage{amsmath,amssymb,amsfonts,amsthm,mathtools}

\usepackage{graphicx}
\graphicspath{{images/}}
\usepackage{booktabs}
\usepackage{multirow}
\usepackage{array}
\usepackage{tabularx}
\usepackage{longtable}
\usepackage{rotating}
\usepackage{float}
\usepackage{tikz}
\usetikzlibrary{shapes.geometric}
\usepackage{xcolor}
\usepackage{enumitem}

\usepackage[numbers,sort&compress]{natbib}
\usepackage{url}

\theoremstyle{definition}

\newtheorem{condition}{Condition}
\theoremstyle{remark}

\usepackage{algorithm}
\usepackage{algpseudocode}
\theoremstyle{plain}
\newtheorem*{theorem*}{Theorem}

\DeclareMathOperator*{\argmin}{arg\,min}

\usepackage{authblk}

\newcommand\equalcontrib[1]{%
  \begingroup
  \renewcommand\thefootnote{}\footnote{#1}%
  \addtocounter{footnote}{-1}%
  \endgroup
}

\begin{document}

\title{Theoretical Properties of Multivariate Random Forest in Feature Selection and its Application to Facial Morphology-Gene Detection}

\author[1]{Yangsheng Wang $^\dagger$}
\author[1]{Samruddhi Thakar$^\dagger$}
\author[1]{Anton Schick}
\author[1]{Guifang Fu\thanks{Corresponding author: gfu@binghamton.edu}}

\affil[1]{%
Department of Mathematics and Statistics, Binghamton University,
Binghamton, New York 13902, USA}

\maketitle
\equalcontrib{$^\dagger$These authors contributed equally to this work.}
\begin{abstract}
This work establishes a theoretical foundation for joint feature selection with multivariate outcomes, positioning the permutation-based variable importance measure (PVIM) of multivariate random forests (MRF) as a principled tool for high-dimensional feature selection. We establish the first consistency guaranty for MRF, showing that it retains all truly \textcolor{black}{influential} features with probability tending to one as the sample size \textcolor{black}{grows to infinity} under mild regularity conditions. Incomplete U-statistics is employed to incorporate three layers of randomness: subsampling of subjects for training each tree, \textcolor{black}{subsampling of features} at each split, and permutation of each feature for the out-of-bag (OOB) samples. Unlike independence-based screening that evaluates each feature in isolation, PVIM \textcolor{black}{is a joint screening approach that} accounts for multicollinearity, nonlinear, high-order interactions, and subject heterogeneity via ensemble aggregation. \textcolor{black}{Moreover}, we demonstrate the practical utility of MRF through a genome-wide association study (GWAS) of human facial morphology (with 2,342 \textcolor{black}{subjects} and 453,273 SNPs), where MRF identifies several novel loci and interaction hubs that extend prior findings. Extensive simulations show that \textcolor{black}{MRF accurately identifies truly influential signals while producing parsimonious feature sets with \textcolor{black}{well controlled} false selection rates}, outperforming canonical correlation analysis (CCA) and several other independence multivariate screening approaches. In addition, we also propose a novel simulation framework, including \textcolor{black}{image outcomes}, that more closely mimic the intricate nature of real-world data and provide rigorous testbeds for machine learning research.

\end{abstract}

\textbf{Keywords:}
Multivariate Random Forests, Permutation-based Variable Importance Measure, Sure Screening, Feature Selection, GWAS

\section{Introduction}
Feature selection has been widely recognized as a fundamental challenge across diverse disciplines, particularly in high-dimensional settings where the number of features far exceeds the number of observations \citep{buhlmann2011statistics}. Selecting a subset that retains only the truly \textcolor{black}{influential} features is essential for enhancing interpretability, reducing computational costs, mitigating overfitting, improving predictive performance, and discovering novel scientific insights \citep{saeys2007review}. Over the past several decades, extensive research has produced a broad spectrum of feature selection techniques \citep{10.5555/944919.944968,li2017feature, BARBIERI2024123667}. Classical feature selection approaches for multivariate outcomes are typically grouped into three major categories: filter methods, wrapper methods, and embedded methods \citep{10.5555/1208773, CHANDRASHEKAR201416}.

Filter methods evaluate each feature individually based on its statistical association strength with the outcome vector, using criteria such as information-theoretic measures or dependence scores \citep{Timme2013, Ding2003}. Among these, sure independence screening (SIS) and its variants have gained widespread attention, offering great computational efficiency and convenience for ultrahigh-dimensional feature selection with theoretical guarantees \citep{fan2008sure}. For instance, \citet{Li2012} proposed distance correlation–based SIS (DC-SIS) to capture nonlinear associations. \citet{Zhao2022} further extended it to a multivariate rank distance correlation approach (MrDc-SIS), achieving model-free and distribution-free properties. \citet{Pan2019} developed BCor-SIS, leveraging ball correlation to enhance robustness under mild distribution assumptions. \citet{Liu2020} combined projection correlation with knockoff features (PC-Screen) to control the false selection rate, while \citet{Guo2022} improved ranking stability by modifying the weighting scheme of distance correlation (SC-SIS). \textcolor{black}{Another well-known filter method is canonical correlation analysis (CCA), 
which constructs linear combinations of features and outcomes that maximize their 
correlation \citep{80ba5479-167b-33ca-b5ea-cdb01c308499,hardoon2004canonical, Härdle2024}. 
In feature screening settings, CCA is typically applied to each feature individually, 
and statistical significance can be assessed using Rao’s $F$-approximation to obtain 
a $p$-value measuring the association strength between each feature and the multivariate 
outcome vector} \citep{johnson2007applied, rencher2012methods}. 
CCA has been widely used in genome-wide association studies (GWAS) to detect 
significant single nucleotide polymorphisms (SNPs) associated with various traits \citep{ferreira2009multivariate,seoane2014canonical, osborne2024multivariate}. In summary, filter methods are attractive for their computational efficiency and scalability in ultrahigh-dimensional settings. They provide practical convenience by relying solely on dependence measures, without requiring explicit model specification. However, this simplicity comes with trade-offs: by evaluating each feature in isolation, filter methods assume independence and thus fail to account for \textcolor{black}{joint effects}, correlations, and interactions among features \citep{saeys2007review,JMLR:v24:21-1046}.

Wrapper methods are iterative, model-based feature selection approaches that evaluate a large number of candidate subsets by comparing the predictive performance of the corresponding models. Common implementations rely on greedy algorithms such as forward selection, backward elimination, or stepwise selection \citep{kohavi1997wrappers, 10.5555/944919.944968}, with evaluation criteria including multivariate $R^2$, mean squared error (MSE), AIC, or BIC, etc \citep{vandenBurg1988}. By exploiting model structure directly, wrapper methods can capture joint effects among features that are overlooked by filter methods \citep{kohavi1997wrappers, 10.5555/944919.944968}. However, exhaustive evaluation of all possible subsets requires $2^p$ combinations, which is infeasible except in low-dimensional settings \citep{blum1997selection}. To mitigate this, more advanced algorithms, such as sequential floating search, branch-and-bound, and evolutionary strategies, have been proposed to explore the subset space more efficiently. Despite these advances, \textcolor{black}{wrapper methods remain computationally intensive and prone to overfitting because repeatedly evaluating a large number of candidate subsets with overlapping features on the same dataset results} in selection bias \citep{CHANDRASHEKAR201416, saeys2007review, kraev2024shapselect, JMLR:v24:21-1046, pudjihartono2022review, reunanen2003overfitting}. Moreover, the features selected \textcolor{black}{by wrapper methods} are tightly coupled to the model employed during selection, which limits their transferability to different structures \citep{saeys2007review, pudjihartono2022review}.

Embedded methods integrate feature selection directly into the model learning process, thereby combining the advantages of both filter and wrapper approaches. Unlike filter methods, which evaluate each feature individually and independently of a predictive model, and wrapper methods, which assess candidate subsets through exhaustive searches, embedded methods perform feature selection as part of the model learning itself. This allows the algorithm to automatically identify and retain \textcolor{black}{influential} features based on their contributions to predictive performance \citep{10.5555/944919.944968}. Compared with wrappers, embedded methods are generally more computationally efficient and thus more scalable to high-dimensional datasets. Compared with filters, embedded methods achieve better predictive accuracy because the importance of each \textcolor{black}{feature is assessed in conjunction with all other features based on model performance on test data}, rather than in isolation. A prominent example is the multivariate LASSO, which applies regularization penalties to shrink the coefficients of weakly associated features toward zero, thereby enabling simultaneous feature selection and improved predictive accuracy \citep{Simon2013, Raess2016}. Despite these advantages, LASSO-based methods \textcolor{black}{are restricted} to specific linear model assumptions. They may struggle to capture complex relationships among features, particularly when the regularization parameters are not carefully tuned.

While existing methods have laid a solid foundation for feature selection, they remain limited in addressing multivariate outcomes, often due to restrictive parametric or linear assumptions, oversimplified model formulations, and independence-based screening strategies that evaluate features in isolation and thus fail to capture complex structures. To overcome these limitations, we \textcolor{black}{propose using} multivariate random forests (MRF), which have several appealing properties: they are model-free and distribution-free, and their hierarchical tree structure naturally captures complex and nonlinear relationships and higher-order interactions among features. Furthermore, by randomly selecting subsets of features at each split, MRF effectively mitigates multicollinearity, and also enhances predictive accuracy by averaging across \textcolor{black}{an} ensemble of several trees. In particular, the permutation-based variable importance measure (PVIM) quantifies each feature’s overall contribution to prediction by evaluating the change in prediction error obtained from the out-of-bag (OOB) sample after that feature is permuted, conditioning on the presence of all other features. Despite these advantages, two critical challenges must be addressed before MRF can serve as a rigorous feature selection framework: establishing theoretical guarantees for PVIM and assessing its computational feasibility in high-dimensional settings.

In this paper, we fill these gaps by establishing an asymptotic consistency screening property for the permutation-based variable importance measure produced by multivariate random forests. In doing so, we provide a framework for joint, interaction-aware, and heterogeneous-aware screening for multivariate outcomes, thereby offering a new theoretical foundation for the use of MRF in high-dimensional feature selection. A feature selection procedure is said to possess the sure screening property if  
\[
P\left( \mathcal{D} \subseteq \widehat{\mathcal{D}}_n \right) \longrightarrow 1 
\qquad \text{as } n \to \infty,
\]  
where $\mathcal{D}$ denotes the set of truly \textcolor{black}{influential} features and $\widehat{\mathcal{D}}_n$ the set of selected features. This theoretical conclusion ensures that, with probability approaching one, the procedure retains all true features as the sample size becomes sufficiently large. This theoretical guarantee is especially important in high-dimensional applications, such as the motivating genome-wide association studies, where it is crucial to ensure that the selected subset contains all truly \textcolor{black}{influential SNPs that are indeed associated with the human facial morphology}.

\textcolor{black}{The motivating large-scale data consists of} high-resolution 3D facial morphology images (each image is represented as 7,160 quasi-landmark 3D point-cloud XYZ coordinates) and 453,273 SNPs (after quality control) measured for each of the 2,342 human subjects. Following the standard procedure described in \citet{Claes2018}, we extract the first 50 principal components (PCs) from the original 3D facial morphology images and utilize them as multivariate outcomes. As a very standard feature selection procedure in this field, CCA was previously applied to this same dataset, where it evaluated the importance of each SNP individually. In contrast to independence feature selection approaches, MRF emphasizes a very different strategy, naturally accounting for feature dependencies and complex structures. Using MRF, we identify several novel genetic associations not previously reported in any prior studies, while also confirming some genes that have been highlighted in the literature using entirely different methodologies, data sources, foci, or study designs. Furthermore, we evaluate the practical performance of MRF as a feature selection approach through extensive simulation studies under various settings designed to reflect \textcolor{black}{the} nonlinear and complex nature of real-world applications. Across all scenarios, our results demonstrate that MRF consistently outperforms existing feature selection methods in detecting truly \textcolor{black}{influential} features with either strong main effects or \textcolor{black}{weak main effects but having strong interactions with other features. } 

The remainder of this paper is organized as follows. Section 2 introduces the methodology and presents the theoretical results. Section 3 details the designs and results of the simulation studies and compares MRF with several relevant approaches. Section 4 demonstrates the application of MRF to the human facial morphology GWAS dataset and highlights the novel discoveries. Section 5 discusses the broader implications of this work and outlines potential directions for future research. Additional technical proof details are provided in the Appendix.

\section{Methodology}
Let $\mathbf Z = (\mathbf X, \mathbf Y) \in \mathbb{R}^{p+q}$ be a random vector, 
where $\mathbf X = (X^{(1)}, \ldots, X^{(p)})^\top \in \mathbb{R}^p$ denotes the 
$p$-dimensional feature vector and 
$\mathbf Y = (Y^{(1)}, \ldots, Y^{(q)})^\top \in \mathbb{R}^q$ denotes the 
$q$-dimensional outcome vector. The \textcolor{black}{entire sample} consists of $n$ \textcolor{black}{independent copies }of $\mathbf Z$, denoted by
\[
\mathbf Z_1, \ldots, \mathbf Z_n.
\]
\textcolor{black}{Within the bagging framework, each tree is constructed using a subsample randomly drawn from the full dataset without replacement.} The ensemble contains $m_n$ decision trees and each tree is trained on a subset of $\{\mathbf{Z}_1, \ldots, \mathbf{Z}_n\}$ of size $k_n$ \textcolor{black}{with unique indices}.  Let $\mathcal{A}_{n,k_n}$ denote the collection of all subsets of \textcolor{black} {indices} $\{1, \ldots, n\}$ of size $k_n$. We draw $m_n$ subsets \textcolor{black}{of indices}, denoted by $S_1, \ldots, S_{m_n}$, \textcolor{black}{at random} with replacement from $\mathcal{A}_{n,k_n}$; that is,
\[
S_1, \ldots, S_{m_n} \overset{\text{iid}}{\sim} \mathrm{Unif}(\mathcal{A}_{n,k_n}).
\]
Furthermore, the subsets $S_1, \ldots, S_{m_n}$ are sampled independently of the training sample $\mathbf{Z}_1, \ldots, \mathbf{Z}_n$. Denote the collection of \textcolor{black}{these subsets realized across all $m_n$ trees} by 
\[
\mathcal{S}_{k_n,m_n} = \{S_{1},...,S_{m_n}\}.
\]
For each subset $S_m$, let
\[
\mathcal{Z}_{S_m} = \{\mathbf{Z}_i : i \in S_m\}
\]
be the associated subsample of the data.

For a multivariate outcome we use the following splitting criterion at each node of each tree. Consider a parent node $A$ containing $n_A$ observations. To split $A$ into two child nodes, 
a left node $A_L$ and a right node $A_R$, let $a$ be a candidate split threshold of feature $X^{(j)}$. 
For a continuous feature $X^{(j)}$, define
\[
A_L(j,a) = \{ i \in A : X_{i}^{(j)} < a \}, 
\qquad 
A_R(j,a) = \{ i \in A : X_{i}^{(j)} \geq a \}.
\]
For a discrete feature $X^{(j)}$, define
\[
A_L(j,a) = \{ i \in A : X_{i}^{(j)} = a \}, 
\qquad 
A_R(j,a) = \{ i \in A : X_{i}^{(j)} \neq a \}.
\]
The within-node error, defined as the sum of squared Mahalanobis distances, is
\[
\mathrm{SSE}(A) \;=\; \sum_{i \in A} 
\bigl(\mathbf{Y}_i - \bar{\mathbf{Y}}^{(A)}\bigr)^\top 
V^{-1}(A) 
\bigl(\mathbf{Y}_i - \bar{\mathbf{Y}}^{(A)}\bigr),
\]
\textcolor{black}{where $V(A)$ denotes the covariance matrix of the outcome vector estimated using the observations in node $A$.}  The average outcome vector at a node $A$ is
\[
\bar{\mathbf{Y}}^{(A)} \;=\; \frac{1}{n_A}\sum_{i \in A} \mathbf{Y}_i.
\] \textcolor{black}{ Let $n_L$ and $n_R$ denote the number of observations in the left child node $(A_L)$ and right child node $(A_R)$, respectively.} The optimal split $(j^*, a^*)$ at a node A is determined by solving the following optimization problem:
$$(j^*, a^*) = \argmin_{j, a} \left[ \frac{n_{L}}{n_A} \mathrm{SSE}\bigl(A_L(j,a)\bigr) + \frac{n_{R}}{n_A} \mathrm{SSE}\bigl(A_R(j,a)\bigr
) \right].$$
\textcolor{black}{The splitting at node $A$ stops if every candidate split produces a child node with fewer than five observations; consequently, every terminal node contains at least five observations.}

In random forests, in addition to the randomness introduced by subsampling the subjects, there is further randomness arising from \textcolor{black}{subsampling the features}. That is, the random forest = bagging + feature randomness. At each split, rather than considering all $p$ features, the algorithm randomly selects a subset of candidate features for splitting, with the subset size determined by the tuning parameter \texttt{mtry} (as commonly denoted in software implementation).

Specifically, \textcolor{black}{let $\mathcal{B}_{p,\texttt{mtry}}$ denote the collection of all subsets of indices $\{1, \ldots, p\}$ of size \texttt{mtry}. For each tree $m$, we draw $k_n^*$ subsets of such indices, denoted by $J_{m1},...,J_{mk_n^*}$, at random with replacement from $\mathcal{B}_{p,\texttt{mtry}}$; that is,
\[
J_{m1},...,J_{mk_n^*} \overset{\text{iid}}{\sim} \mathrm{Unif}(\mathcal{B}_{p,\texttt{mtry}}).
\] 
Denote the collection of all such random subsets realized across all internal nodes for the $m^{\text{th}}$ tree by 
\[
\omega_m = \{J_{m1}, \dots, J_{mk_n^*}\},
\]}
which captures the entire randomness in the selection of candidate features at all internal nodes for constructing each tree, and is selected prior to growing a tree. Here, $k_n^*$ denotes an upper bound on the number of possible \textcolor{black}{ internal nodes} in each tree. This feature subsampling procedure is performed independently at every internal node of every tree. Although a realized tree may contain fewer than $k_n^*$ internal nodes, we may \textcolor{black}{generate $\omega_m$} in advance and apply it sequentially to all the internal nodes during the recursive splitting process. \textcolor{black}{Any unused components of $\omega_m$} are simply discarded. An illustrative example of this mechanism is provided in 
Figure~\ref{node-feature-selection}. Across trees, these random collections $\omega_{1}, \ldots, \omega_{m_n} $ are  independent and identically distributed, and are independent of the training set $(\mathbf{Z}_1,...,\mathbf{Z}_n)$ and the  subject-subsampling set $\mathcal S_{k_n,m_n}$. 
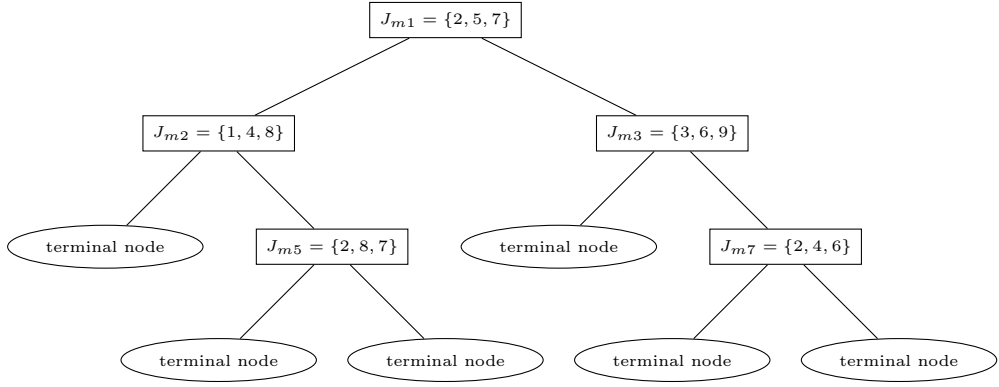
\begin{figure}[h]
\centering
\begin{tikzpicture}[
  level distance=1.5cm,
  level 1/.style={sibling distance=6cm},
  level 2/.style={sibling distance=3cm},
  every node/.style={draw, font=\tiny}
]

\node { {$J_{m1} = \{2,5,7\}$} }
    child {
        node { {$J_{m2} = \{1,4,8\}$} }
            child {
                node[ellipse] { {terminal node} }
            }
            child {
                node { {$J_{m5} = \{2,8,7\}$} }
                    child {
               node[ellipse] { {terminal node} }
            } 
            child {
                node[ellipse] { {terminal node} }
            }
            }
    }
    child {
        node { {$J_{m3} = \{3,6,9\}$} }
            child {
                node[ellipse] { {terminal node} }
            }
            child {
                node { {$J_{m7} = \{2,4,6\}$} } 
                 child {
                node[ellipse] { {terminal node} }
            } 
            child {
                  node[ellipse] { {terminal node} }
            }}
            }
    ;

\end{tikzpicture}
\caption{Illustration of node-wise feature subsampling in the $m^{\text{th}}$ tree. Each internal node is assigned one component of $\omega_m$ of size \texttt{mtry} = 3 in a sequence. The splitting stops at terminal node so $J_{m4}$ and $J_{m6}$ remain unused.}
\label{node-feature-selection}
\end{figure}

Throughout the article, \textcolor{black}{we use $T_{k_n}(\mathcal{Z}_{S_m},\omega_m)$ to denote the \textcolor{black}{estimate trained from the $m^{\text{th}}$ decision tree of the MRF built on subsampling of subjects reflected by $\mathcal{Z}_{S_m}$ and subsampling of features governed by $\omega_m$.} For a new observation $\mathbf{X}^*$, the corresponding prediction from this tree is denoted by
$T_{\mathbf{X}^*, k_n}(\mathcal{Z}_{S_m}, \omega_m).$}
For a new observation with features $\mathbf{X}^*$, the MRF prediction is given by
\[
 \frac{1}{m_n} \sum_{m=1}^{m_n} 
T_{\mathbf{X}^*, k_n}(\mathcal{Z}_{S_m}, \omega_m).
\]

\subsection{Permutation Variable Importance Measure of the Multivariate Random Forests}
After subsampling, each tree has $(n-k_n)$ OOB observations. The PVIM quantifies the change in the OOB prediction error caused by permuting a single feature while keeping all other features fixed. If $X^{(j)}$ is \textcolor{black}{influential}, this permutation breaks its original association with the outcome and increases the prediction error, yielding a larger PVIM. 

\textcolor{black}{Specifically,  let $\{\pi_{mj}, 
m = 1, \ldots, m_n, 
j = 1, \ldots, p \}$
denote permutations of $n - k_n$
elements, each drawn independently and uniformly at random from the set of all non-identity permutations of $ n - k_n$
elements.} \textcolor{black}{After training the $m^{\text{th}}$ tree, the permutation $\pi_{mj}$ is applied to the indices of the out-of-bag observations for each feature $j$ independently. That is,}
for $i \in S_m^C$, we define
\[
\widetilde{\mathbf{X}}^{(j,\pi_{mj})}_i 
= \bigl(
X_i^{(1)}, \dots, X_i^{(j-1)},\,
X_{\pi_{mj}(i)}^{(j)},\,
X_i^{(j+1)}, \dots, X_i^{(p)}
\bigr),
\]
where the $j^{\text{th}}$ feature of the $i$-th out-of-bag observation 
for tree $m$ is permuted according to a permutation $\pi_{mj}$ while
other features remain unchanged. 
For each tree $m$, the
 prediction difference before and after permutation of feature $X^{(j)}$ obtained from the tree $T_{k_n}(\mathcal{Z}_{S_m}, \omega_m)$ is
\begin{equation} \label{vimp for a tree}
\begin{split} 
I_j(\mathcal{Z}_{S_m}, \omega_m, \pi_{mj}) 
&= \frac{1}{n-k_n} \sum_{i \in S_m^C} 
   \left\| \mathbf{Y}_i - T_{\Tilde{\mathbf{X}}^{(j,\pi_{mj})}_i,k_n}(\mathcal{Z}_{S_m}, \omega_m) \right\|^2 \\[6pt]
&\quad - \frac{1}{n-k_n} \sum_{i \in S_m^C} 
   \left\| \mathbf{Y}_i - T_{\mathbf{X}_i,k_n}(\mathcal{Z}_{S_m}, \omega_m) \right\|^2.
\end{split}
\end{equation}



Finally, the PVIM estimator based on MRF is given by,
\begin{equation}\label{vimp}
    \widehat{\text{PVIM}}(X^{(j)}) 
    = \hat{\Lambda}_n^{(j)}
    = \frac{1}{m_n} \sum_{m=1}^{m_n} I_j(\mathcal{Z}_{S_m}, \Omega_{mj})
\end{equation}
 \textcolor{black}
{where, $\Omega_{mj} = (\omega_{m}, \pi_{mj})$ for $j=1,...,p$ and $m =1,...,m_n$.}
\noindent
\textcolor{black}{This estimator approximates the population-level PVIM of the MRF, given by,}
\begin{equation} \label{lambda}
   \text{PVIM}(X^{(j)}) = \Lambda^{(j)}
= \mathbf{E}\Big[I_j(\mathcal{Z}_{S_m}, \Omega_{mj})\Big].
\end{equation}



\subsection{Feature Screening Procedure Based on PVIM and Its Theoretical Properties}
 
In this section, we claim that the PVIM is consistent and possesses the sure screening property under certain regularity conditions. 
In high-dimensional settings, the number of features typically satisfies $p \gg n$, while only a small subset of them truly influences the outcome. Accordingly, we adopt an asymptotic framework in which $p \to \infty$ as $n \to \infty$, whereas the outcome dimension $q$ remains fixed. The cardinality of the subset containing all truly influential features is allowed to grow with $n$ \textcolor{black}{but at a substantially slower rate.}

Consider the conditional distribution of $\mathbf{Y}$ given $\mathbf{X}$, denoted as $Q_{\mathbf{X}}(B) = \mathrm{P}(\mathbf{Y} \in B \, | \, \mathbf{X}) $, where $B$ represents any Borel set of $\mathbb{R}^q$. 
We define the index sets of all truly influential features as,
\begin{equation} \label{D}
    \mathcal{D} = \text{ \{$j \in \{ 1,...,p\}$: $Q_\mathbf{X}(B)$  depends on $X^{(j)}$ for  at least one Borel set $B$\}}, 
\end{equation}
which indicates that, $Q_{\mathbf{X}}(B)\overset{a.s}{=} Q_{\mathcal{X}_{\mathcal{D}}}(B)$ for all Borel sets $B$ of $\mathbb{R}^q$, where \[\mathcal{X}_{\mathcal{D}} = \{ X^{(j)} \, | \, j \in \mathcal{D}\}.\]  For the development of theoretical results, we impose the following regularity condition:

\begin{condition} \label{threshold}
    There exists a \textcolor{black} {constant $c_0>0$} such that for every $j \in \mathcal{D}$, $\Lambda^{(j)}> c_0$.
\end{condition}
The set of indices of important features selected by MRF is determined based on their estimated PVIM values as
\begin{equation} \label{D hat star}
  \widehat{\mathcal{D}}_n =\{j \in \{1,...,p\}: \hat{\Lambda}_n^{(j)}> c_0/2 \}, 
\end{equation}
where $c_0$ is the threshold specified in Condition~\ref{threshold}.


\textcolor{black}{For each tree $m=1,\dots,m_n$ with subsample 
$S_m \subset \{1,\dots,n\}$, we  denote the prediction error for an observation $\mathbf{Z}_i$ as}
\begin{equation*}\label{e}
g_n(\mathbf{Z}_i,\mathcal{Z}_{S_m},\omega_m)
=
\big\|
\mathbf{Y}_i -
T_{\mathbf{X}_i,k_n}(\mathcal{Z}_{S_m},\omega_m)
\big\|^2,
\end{equation*}
where $i=1,\ldots,n$.
Moreover,  \textcolor{black}{we denote the prediction} error obtained after permuting the $j$-th feature according to $\pi_{mj}$ as
\begin{equation*}\label{e_j}
g_{nj}(\mathbf{Z}_i,\mathcal{Z}_{S_m},\Omega_{mj})
=
\big\|
\mathbf{Y}_i -
T_{\tilde{\mathbf{X}}_i^{(j,\pi_{mj})},k_n}
(\mathcal{Z}_{S_m},\omega_m)
\big\|^2.
\end{equation*}
With this notation in place, observe that, for each $j=1,\dots,p$, the permutation variable importance estimator $\hat{\Lambda}_n^{(j)}$ can be expressed as
\[
\hat{\Lambda}_n^{(j)} 
=
\frac{1}{m_n}\frac{1}{n-k_n}
\left(
\sum_{m=1}^{m_n}\sum_{i=1}^{n} 
g_{nj}(\mathbf{Z}_i,\mathcal{Z}_{S_m},\Omega_{mj})\mathbb{1}_{\{i\notin S_m\}}
-
\sum_{m=1}^{m_n}\sum_{i=1}^{n} 
g_{n}(\mathbf{Z}_i,\mathcal{Z}_{S_m},\omega_m)\mathbb{1}_{\{i\notin S_m\}}
\right).
\]
Finally, define
\begin{align}
\bar g_n(\mathbf{Z}_i,\mathcal{Z}_{S_m})
&=
\mathbf{E}\Big[
g_n(\mathbf{Z}_i,\mathcal{Z}_{S_m},\omega_m)
\,\big|\,
\mathbf{Z}_1,\dots,\mathbf{Z}_n,\mathcal S_{k_n,m_n}
\Big],  \notag \\
 \text{and} \quad \bar g_{nj}(\mathbf{Z}_i,\mathcal{Z}_{S_m}) \notag
&=
\mathbf{E}\!\Big[
g_{nj}(\mathbf{Z}_i,\mathcal{Z}_{S_m},\Omega_{mj})
\,\big|\,
\mathbf{Z}_1,\dots,\mathbf{Z}_n,\mathcal S_{k_n,m_n}
\Big], 
\qquad j=1,\dots,p. \notag \label{bar_e_j}
\end{align}

\noindent
The following theorem is stated in a setting where all observations $\mathbf{Z}_1,...,\mathbf{Z}_n$ as well as number of features $p$, the threshold $c_0$ defined in Condition (\ref{threshold}) and the set of true features $\mathcal{D}$ may vary with $n$, i.e $\mathbf{Z}_1 = \mathbf{Z}_{1,n},....,\mathbf{Z}_n = \mathbf{Z}_{n,n}$, $p = p_n$, $c_0 = c_{0n}$ and $\mathcal{D} = \mathcal{D}_n$. 
To keep the notation simple we suppressed this dependence on $n$.

\begin{theorem*}[Consistency and Sure Screening]\label{thm:sure-screening}
Suppose the following holds for some $\alpha \ge 0$:

\begin{itemize}

\item[(C1)]
$\mathbf{E}\!\Big[
\Bigl(
g_n(
\mathbf{Z}_{k_n+1},
\mathbf{Z}_1, \ldots, \mathbf{Z}_{k_n},
\omega_1
)
-
\bar{g}_n(
\mathbf{Z}_{k_n+1},
\mathbf{Z}_1, \ldots, \mathbf{Z}_{k_n}
)
\Bigr)^2
\Big]
= O(n)$ and,

$\max\limits_{1 \le j \le p}\mathbf{E}\!\Big[
\Bigl(
g_{nj}(
\mathbf{Z}_{k_n+1},
\mathbf{Z}_1, \ldots, \mathbf{Z}_{k_n},
\Omega_{1j}
)
-
\bar{g}_{nj}(
\mathbf{Z}_{k_n+1},
\mathbf{Z}_1, \ldots, \mathbf{Z}_{k_n}
)
\Bigr)^2
\Big]
= O(n), $ 

\item[(C2)]
$\mathbf{E}\!\Big[
\bar{g}_n^2\!\left(
\mathbf{Z}_{k_n+1},
\mathbf{Z}_1, \ldots, \mathbf{Z}_{k_n}
\right)
\Big]
= O\!\left(k_n^{\alpha}\right)$ and,

$\max\limits_{1 \le j \le p}\mathbf{E}\!\left[
\bar{g}_{nj}^2\!\left(
\mathbf{Z}_{k_n+1},
\mathbf{Z}_1, \ldots, \mathbf{Z}_{k_n}
\right)
\right]
= O\!\left(k_n^{\alpha}\right), $
\item[(C3)]  $\underset{n \to \infty}{\lim} \dfrac{n}{m_n} =0$ \, and \, $\underset{n \to \infty}{\lim} \dfrac{k_n^{2+\alpha}}{n} =0$.
\end{itemize}
Then 
\begin{equation} \label{consistency}
   \underset{1 \leq j \leq p}{\max} \mathrm{Var}(\hat{\Lambda}_n^{(j)}) = O(R_n),
\end{equation}
where $R_n = \max\Big\{\dfrac{k_n^{2+\alpha}}{n}, \dfrac{n}{m_n}\Big\}$.
\textcolor{black}{Furthermore, assume that 
\begin{equation}\label{*}
 \quad \underset{n \to \infty}{\lim}\dfrac{ \mathrm{Card}(\mathcal{D})R_n}{c_0^2} = 0.  
\end{equation}}
Then as long as Condition \ref{threshold} is satisfied we have the sure screening property, that is,
\begin{equation} \label{ss equation}
   \mathrm{P}\!\left(
\mathcal{D} \subset \widehat{\mathcal{D}}_n
\right)
\longrightarrow 1
\quad \text{as} \,\, n \to \infty, 
\end{equation}
where $\mathcal{D}$ and $\widehat{\mathcal{D}}_n$ are defined in 
\eqref{D} and \eqref{D hat star}, respectively.
\end{theorem*}

{\color{black}The existing sure screening literature typically assumes that the feature dimension \(p\), the set of true features \(D\), and the threshold \(c_0\) are fixed and do not vary with \(n\). In contrast, our theorem is stated in a more general setting, allowing these parameters to diverge with \(n\). The classical fixed-parameter setting is therefore a direct special case of our result, as stated in the following remark.}
\vspace{2em}

\noindent
\textcolor{black}{\textbf{Remark.} If the number of features \(p\), \(\mathrm{Card}(D)\), and \(c_0\) are fixed and do not vary with \(n\), then assumption~(\ref{*}) reduces to \(\underset{n \to \infty}{\lim} R_n = 0.\) In this case the sure screening property still holds
\[
\mathrm{P}\!\left(
\mathcal{D} \subset \widehat{\mathcal{D}}_n
\right)
\longrightarrow 1
\quad \text{as } n \to \infty.
\]}

 \noindent
The proof of the theorem is provided in the Appendix.

\section{Simulation Studies}
\subsection{Evaluation Metrics}
In this section, we evaluate the finite-sample performance of MRF using three simulation examples and compare it with several other feature selection methods, including multivariate LASSO, CCA, and five SIS approaches developed for multivariate outcomes: DC-SIS, BCor-SIS, PC-Screen, SC-SIS, and MrDc-SIS. We assess the performance of each feature selection method using the following four criteria across 100 or 10 replicated simulation datasets:
\begin{itemize}
\item \textbf{Minimum Selection Size ($S$)}: The smallest number of features required to ensure that all true features are retained. Boxplots of $S$ across all replicated simulation datasets are used to evaluate the accuracy and efficiency of selection performance. Smaller mean values or lower box locations reflect \textcolor{black}{higher selection accuracy}, while smaller standard errors or narrower box widths indicate higher efficiency in selecting the true features.

\item \textbf{Individual Success Rate ($P_s$)}: The proportion of times each true feature is successfully selected across all replicated simulation datasets, according to the respective thresholding rule of each method. For CCA, the threshold is determined using a Bonferroni-adjusted family-wise error rate, i.e., $0.05/p$. LASSO automatically selects all features with nonzero coefficients, and its tuning parameter $\lambda$ is chosen by minimizing the cross-validation error. For the five multivariate SIS approaches and the MRF, features are ranked in decreasing order of their respective dependence measure (i.e., dependence score for SIS and PVIM for MRF), and then a max-ratio rule is applied to determine the selection threshold. Specifically, we retain features that fall within the intersection of two sets: (i) the top 5\% of features ranked by their dependence measure and (ii) all features ranking within the top five max-ratio thresholds. This dual criterion ensures that the selected features exhibit both absolute and relative importance. Details of the max-ratio rule are provided in \citet{Zhao2022}.

\item \textbf{Overall Success Rate ($P_a$)}: The proportion of times all truly \textcolor{black}{influential} features are simultaneously selected across all replicated simulation datasets, according to the respective thresholding rule of each method. \textcolor{black}{Higher $P_s$ and $P_a$ values indicate a stronger ability of the method to discover the true features.}

\item \textbf{False Selection Rate ($T$)}: The ratio of the number of noise features incorrectly selected according to the thresholding rule of each method to the total number of noise features. Boxplots of $T$ across all replicated simulation datasets are used to evaluate the robustness and accuracy of selection performance. Smaller values of $T$ indicate \textcolor{black}{fewer incorrectly selected noise features.}
\end{itemize}

\subsection{The Simulation Design}
In the current literature, simulation designs typically generate data from a single parametric and homogeneous model, with randomness introduced only through noise. However, real-world data are often heterogeneous and highly complex. To better reflect these characteristics, we propose two novel simulation frameworks that provide realistic and challenging testbeds for evaluating feature selection methods in general. These designs enable a more rigorous assessment of methodological performance in machine learning literature. In general, methods that are resilient to heterogeneity tend to generalize better to new subjects.

Our simulations incorporate multiple layers of heterogeneity, complexity, and randomness. Firstly, we do not impose any specific model structure. Secondly, we allow different outcome components of different subjects to be influenced by different subsets of features. For example, subject 1’s first two outcomes $(Y_1^{(1)}, Y_1^{(2)})$ may depend on the interaction $X_1^{(10)}\times X_1^{(12)}$, whereas subject 2’s second and fourth outcomes $(Y_2^{(2)}, Y_2^{(4)})$ may be influenced by the three-way interaction $X_2^{(1)}\times X_2^{(3)}\times X_2^{(5)}$. \textcolor{black}{Third, we introduce subject-level threshold to increase heterogeneity. Specifically, 
although $X^{(1)}$ remains associated with $Y^{(1)}, \ldots, Y^{(4)}$, more association variability is induced depending on whether 
$X_i^{(1)} < 5$ or $X_i^{(1)} \ge 5$. }

\subsubsection{The simulation design I}
\textcolor{black}{In Simulation Design I, the feature vector $\mathbf{X} = (X^{(1)}, \ldots, X^{(p)})^\top$  is generated from a multivariate normal distribution with mean vector $\mathbf{0}$ and covariance matrix $\Sigma_{p \times p} = (\sigma_{ij})$, where 
$\sigma_{ij} = 0.8^{|i-j|}$. A set of true features is 
predetermined from $\mathbf{X}$ and denoted by $\mathcal{D}$, while the remaining features in 
$\mathbf{X}$ are treated as noise features. Next, we generate $2q$ independent outcome component mean values from normal distributions, with their means sampled uniformly from the interval $[0,5]$ and their variances sampled uniformly from $[1,10]$. These $2q$ outcome component mean values are then divided into two groups: $\mu_l = \{\mu^{(1)}_l, \ldots, \mu^{(q)}_l\}$ for the left nodes and $\mu_r = \{\mu^{(1)}_r, \ldots, \mu^{(q)}_r\}$ for the right nodes. Each true feature in $\mathcal{D}$ is set to be connected to four outcome component mean values (overlap is allowed). To move beyond the standard parametric and homogeneous simulation framework, we 
construct tree-based structures in which each branch randomly selects an interaction 
order, ranging from one-way up to five-way interactions, while allowing features to 
be sampled with replacement at each internal node to introduce additional nonlinearities. For simplicity, 
subjects are assigned to the left branch if their observations are less than $0.5$. Each subject is then linked to four outcome component mean values according to its involved features along the branch path from the root to the terminal node reached by that 
subject: four outcome component mean values from $\mu_l$ if the subject locates in the left branch, 
or the corresponding four from $\mu_r$ if it locates in the right branch. The observation of each outcome 
component for each subject is generated as the average of the corresponding 
assigned outcome component mean values plus an independent error term 
$\varepsilon \sim N(0,1)$. To further increase complexity, this procedure is repeated multiple times with different tree structures. The final multivariate outcome observations for each subject, 
$\mathbf{Y}_i = (Y^{(1)}_i, \ldots, Y^{(q)}_i)$, is obtained by averaging the 
observations of each outcome component across all trees. This design induces strong correlations among outcome components, as they are 
jointly influenced by overlapping subsets of true features through different 
subjects across multiple trees with different interaction structures. In addition, substantial heterogeneity 
is introduced because different subjects may follow distinct outcome–feature 
relationship patterns. The resulting dataset consists of the outcome matrix 
$\mathbf{Y} \in \mathbb{R}^{n \times q}$ and the feature matrix 
$\mathbf{X} \in \mathbb{R}^{n \times p}$, which are analyzed as if the data-generating mechanism were unknown, mirroring real-world practice.}

\subsubsection{The simulation design II}
In the simulation design II, we preserve the same levels of heterogeneity, complexity, and randomness as in design I when connecting outcomes with true features. However, we further increase the difficulty by replacing the random outcome vectors with 3D images. This design closely mimics the GWAS human facial morphology dataset that motivates our study.

\begin{figure}[h!]
    \centering
    \includegraphics[scale=0.4]{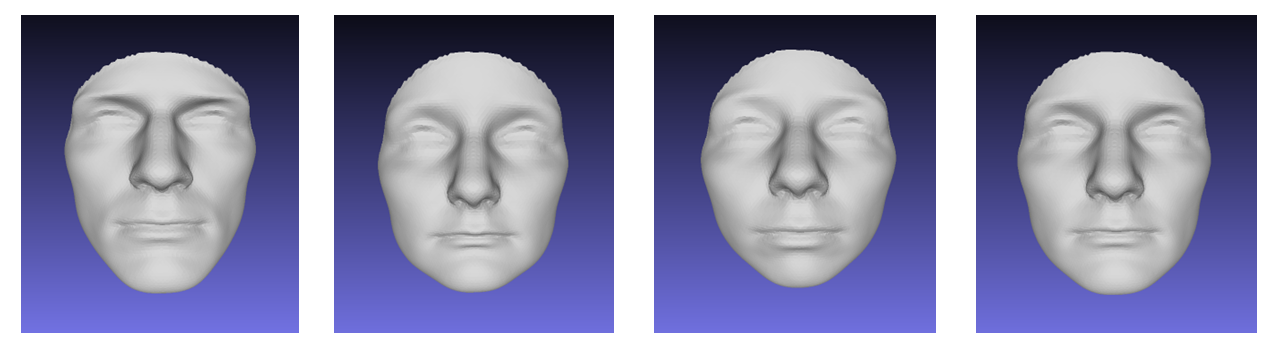}
    \vspace{-1em} 
    \caption{Four examples of synthetic human face simulations.}
    \label{fig:Fig.f}
\end{figure}

\textcolor{black}{To simulate artificial 3D facial images, we first apply principal component analysis (PCA) to the entire real dataset to obtain all principal component (PC) scores. We then randomly select one reference facial morphology and take its first $q$ principal component scores, each perturbed with Gaussian noise to generate $2q$ values, which are partitioned into two groups of $q$ corresponding to the left and right outcome component mean values. Their associations with the true features are constructed following exactly the same framework used in Simulation Design I. From there, five new PC scores are generated for each of the $n$ subjects, serving the same role as the final outcome observations, $\mathbf{Y}_i = (Y^{(1)}_i, \ldots, Y^{(q)}_i)$, for each subject $i$. These five PC scores are retained as the multivariate outcome vector for subsequent analysis. To reconstruct 3D facial images for demonstration purposes, the simulated PC scores are combined with the remaining PC scores of the same reference facial morphology (the one randomly selected earlier to obtain the original PC scores) and transformed back into the facial morphology space to generate $n$ artificial 3D facial images (see Figure~\ref{fig:Fig.f} for examples). The first five simulated PCs explain more than $98\%$ of the total variation in the simulated facial morphologies.}

An important advantage of using PCA to simulate morphologies in this simulation design is that it allows us to regulate and control the dominant modes of structural variation in the morphologies, rather than simply adding pixel-level noise, \textcolor{black}{which would produce images dominated by random pixel fluctuations with no meaningful structure.} Consequently, this design provides a biologically plausible and realistic testbed for evaluating feature selection methods when applied to complex morphology data.

\subsection{Simulation Results}
\subsubsection{Simulation Study 1}
In this simulation, we generate data following the framework described in Simulation Design I. The sample size is set to $n = 200$, with $p = 500$ features and $q = 10$ outcome components. The set of true features is specified as $\mathcal{D} = \{X_1, X_{101}, X_{201}, X_{301}, X_{401}\}$, while the remaining $495$ features are treated as noise. 

\begin{figure}[ht!]
    \centering
    \includegraphics[scale=0.3]{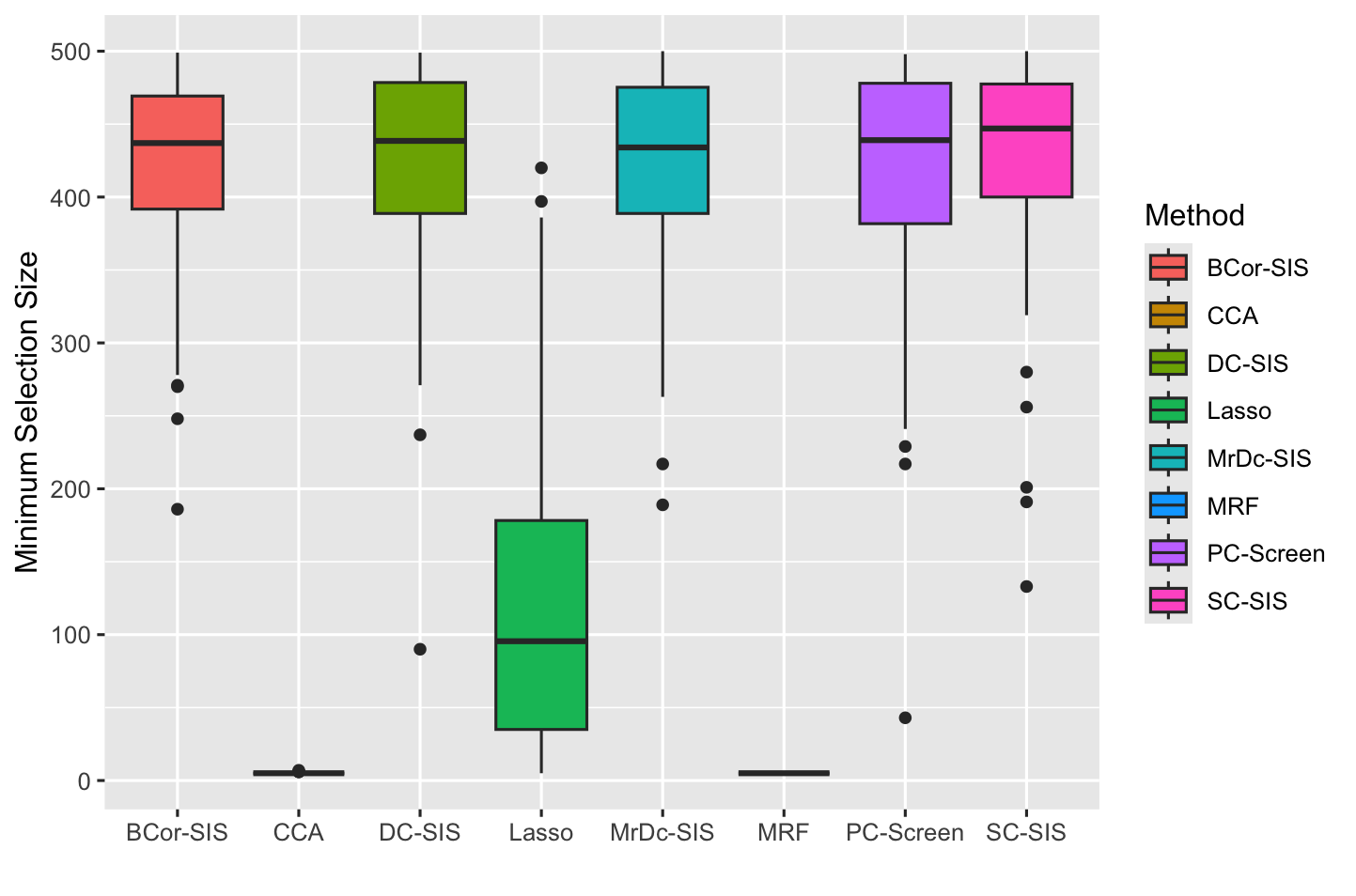}
    \vspace{-1.5em} 
    \caption{Boxplots of the minimum selection size $S$ in Simulation Study 1 across 100 replicated simulation datasets. The eight colors represent the eight approaches. }
    \label{fig:Fig.1}
\end{figure}

\begin{table}[h!]
\caption{
The minimum selection size $S$ for Simulation Study 1. The first row shows the mean of $S$, and the second row shows the standard error of $S$ across 100 replicated simulation datasets.
}
\label{tab:Table.1}
\centering
\begin{tabular}{ccccccccc}
\hline
         & Bcor-SIS & CCA   & DC-SIS & LASSO  & MrDc-SIS & PC-Screen & SC-SIS & MRF \\ \hline
Mean of $S$ & 420.13   & 17.65 & 422.73 & 152.87 & 421.35   & 417.12    & 427.27 & \textbf{7.28} \\
SE of $S$  & 66.93    & 15.69 & 68.16  & 117.77 & 67.23    & 76.96     & 70.30  & \textbf{4.59} \\ \hline
\end{tabular}
\end{table}

Table \ref{tab:Table.1} and Figure \ref{fig:Fig.1} summarize the minimum selection size obtained by eight approaches across 100 replicated simulation datasets in Simulation Study 1. Both the mean and variability of $S$ are reported. As demonstrated in Figure \ref{fig:Fig.1}, MRF and CCA exhibit an overwhelmingly superior performance. Among all the eight approaches, MRF achieves the best overall results, successfully identifying all five true features with an average minimum selection size of only 7.28 out of 500 features, approximately 1/60 of that required by the five SIS approaches (see Table \ref{tab:Table.1}). CCA attains an average minimum selection size of 17.65, slightly larger than MRF with a standard error three times higher (15.69 versus 4.59), but it remains highly competitive among the eight approaches.

\begin{table}[h!]
\caption{
Individual success rate ($P_s$) and overall success rate ($P_a$) for eight methods in the Simulation Study 1.
}
\label{tab:Table2}
\centering
\scalebox{0.75}{
\begin{tabular}{|ccccc|c|ccccc|c|ccccc|c|}
\hline
\multicolumn{6}{|c|}{\textbf{Bcor-SIS}} & \multicolumn{6}{c|}{\textbf{CCA}} & \multicolumn{6}{c|}{\textbf{DC-SIS}} \\ \hline
\multicolumn{5}{|c|}{$P_s$} & $P_a$ & \multicolumn{5}{c|}{$P_s$} & $P_a$ & \multicolumn{5}{c|}{$P_s$} & $P_a$ \\ \hline
$X_1$ & $X_{101}$ & $X_{201}$ & $X_{301}$ & $X_{401}$ & All &
$X_1$ & $X_{101}$ & $X_{201}$ & $X_{301}$ & $X_{401}$ & All &
$X_1$ & $X_{101}$ & $X_{201}$ & $X_{301}$ & $X_{401}$ & All \\
0.01 & 0.05 & 0.03 & 0.03 & 0.01 & 0.00 &
1.00 & 1.00 & 1.00 & 0.74 & 0.98 & 0.74 &
0.01 & 0.05 & 0.03 & 0.01 & 0.05 & 0.00 \\ \hline

\multicolumn{6}{|c|}{\textbf{LASSO}} & \multicolumn{6}{c|}{\textbf{MrDc-SIS}} & \multicolumn{6}{c|}{\textbf{PC-Screen}} \\ \hline
\multicolumn{5}{|c|}{$P_s$} & $P_a$ & \multicolumn{5}{c|}{$P_s$} & $P_a$ & \multicolumn{5}{c|}{$P_s$} & $P_a$ \\ \hline
$X_1$ & $X_{101}$ & $X_{201}$ & $X_{301}$ & $X_{401}$ & All &
$X_1$ & $X_{101}$ & $X_{201}$ & $X_{301}$ & $X_{401}$ & All &
$X_1$ & $X_{101}$ & $X_{201}$ & $X_{301}$ & $X_{401}$ & All \\
1.00 & 1.00 & 1.00 & 1.00 & 0.98 & 0.98 &
0.01 & 0.05 & 0.02 & 0.01 & 0.04 & 0.00 &
0.03 & 0.06 & 0.03 & 0.01 & 0.04 & 0.00 \\ \hline

\multicolumn{6}{|c|}{\textbf{SC-SIS}} & \multicolumn{6}{c|}{\textbf{MRF}} & \multicolumn{6}{c}{} \\ \cline{1-12}
\multicolumn{5}{|c|}{$P_s$} & $P_a$ & \multicolumn{5}{c|}{$P_s$} & $P_a$ & \multicolumn{6}{c}{} \\ \cline{1-12}
$X_1$ & $X_{101}$ & $X_{201}$ & $X_{301}$ & $X_{401}$ & All &
$X_1$ & $X_{101}$ & $X_{201}$ & $X_{301}$ & $X_{401}$ & All &
\multicolumn{6}{c}{} \\
0.00 & 0.02 & 0.00 & 0.06 & 0.02 & 0.00 &
1.00 & 1.00 & 1.00 & 0.84 & 0.99 & \textbf{0.84} &
\multicolumn{6}{c}{} \\ \cline{1-12}
\end{tabular}
}
\end{table}

\begin{figure}[h!]
    \centering
    \includegraphics[scale=0.2]{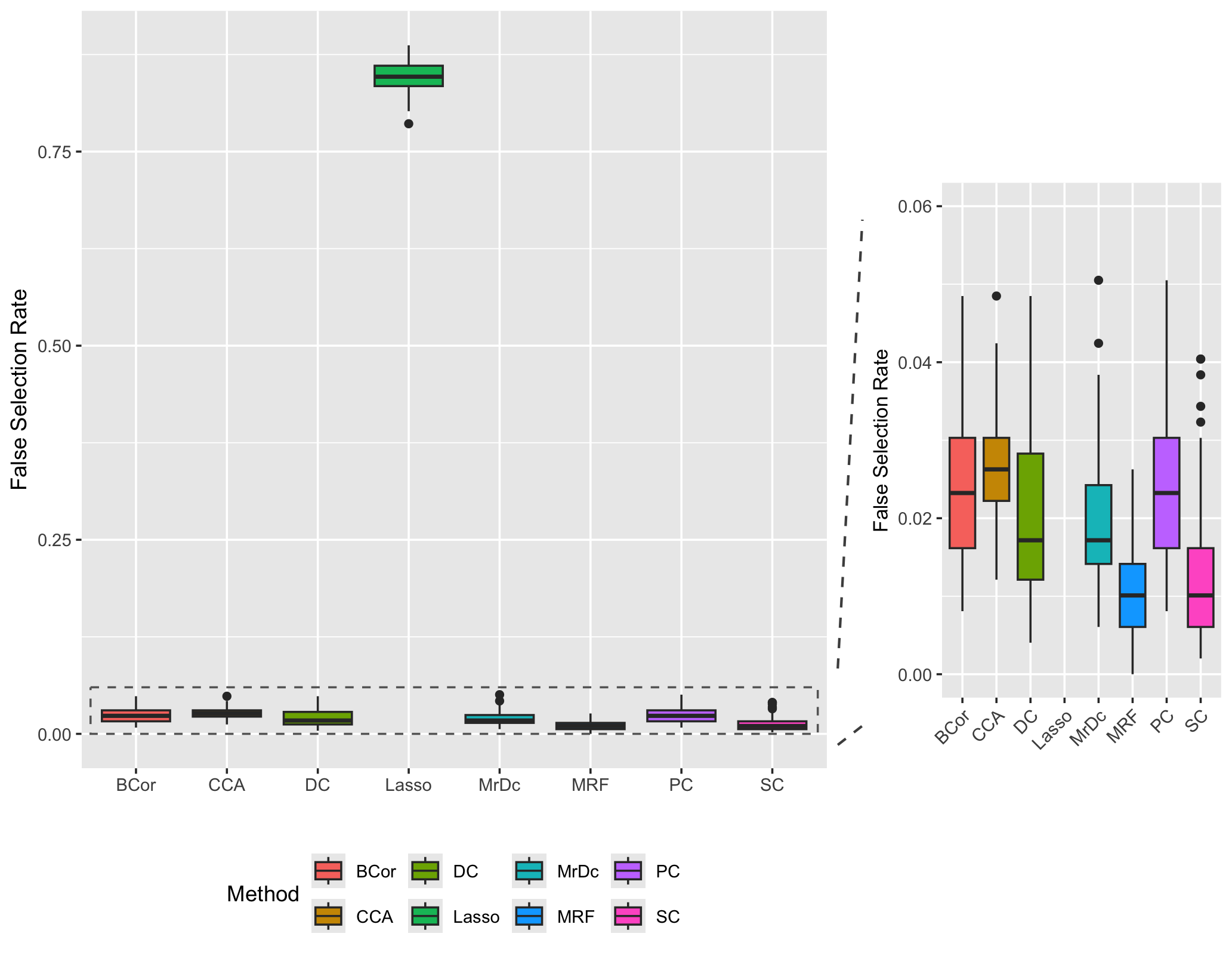}
    \vspace{-1em} 
    \caption{Boxplots of the false selection rate $T$ for Simulation Study 1 across 100 replicated datasets. Left panel: Results for all eight approaches (each color corresponds to one method). Right panel: Zoomed-in results for the seven approaches excluding LASSO.}
    \label{fig:Fig.2}
\end{figure}

Table~\ref{tab:Table2} reports the individual and overall success rates across the eight methods. When applying the max-ratio thresholding rule as MRF does, all five SIS methods fail completely, each yielding an overall success rate of 0\%. In contrast, MRF consistently identifies all true features with an overall success rate of 84\%, while CCA achieves 74\%. LASSO attains the highest overall success rate at 98\%, but this result must be interpreted with caution, as success rates should be evaluated alongside false selection rates. LASSO selects a large number of irrelevant noise features, with an average minimum selection size of approximately 152.87, leading to a substantially inflated false selection rate. As shown in the left panel of Figure~\ref{fig:Fig.2}, LASSO’s false selection rate exceeds 75\% in every replication, with a median above 80\%. By comparison, the right panel of Figure~\ref{fig:Fig.2} demonstrates that MRF achieves a median false selection rate of around 1\% and remains below 3\% across all 100 replicated datasets.

To further assess the performance of the three most promising approaches—MRF, CCA, and LASSO—we consider a more challenging scenario with $p = 10{,}000$ categorical features. To reduce computational burden, we conduct 10 replicated simulations. Table~\ref{tab:Table3} shows that both LASSO and MRF achieve overall success rates ($P_a$) of 90\%, whereas CCA performs substantially worse, with $P_a$ dropping to 50\%. However, LASSO’s high success rate is offset by an extremely large false selection rate, approximately 80-fold higher than that of MRF on average. For brevity, the corresponding tables and figures for this additional simulation are omitted, as they exhibit patterns similar to those already presented.

\begin{table}[h!]
\caption{
Individual success rate ($P_s$) and overall success rate ($P_a$) for the three most promising approaches under the categorical feature setting with high dimension ($p = 10,000$) in the Simulation Study 1.
}
\label{tab:Table3}
\centering
\scalebox{0.9}{
\begin{tabular}{|ccccc|c|ccccc|c|}
\hline
\multicolumn{6}{|c|}{\textbf{LASSO}} & \multicolumn{6}{c|}{\textbf{CCA}} \\ \hline
\multicolumn{5}{|c|}{$P_s$} & $P_a$ & \multicolumn{5}{c|}{$P_s$} & $P_a$ \\ \hline
$X_1$ & $X_{101}$ & $X_{201}$ & $X_{301}$ & $X_{401}$ & All &
$X_1$ & $X_{101}$ & $X_{201}$ & $X_{301}$ & $X_{401}$ & All \\
1.00  & 1.00      & 1.00      & 1.00      & 0.90      & 0.90 &
1.00  & 1.00      & 1.00      & 0.90      & 0.50      & 0.50 \\ \hline
\multicolumn{6}{|c|}{\textbf{MRF}} & \multicolumn{6}{c}{} \\ \cline{1-6}
\multicolumn{5}{|c|}{$P_s$} & $P_a$ & \multicolumn{6}{c}{} \\ \cline{1-6}
$X_1$ & $X_{101}$ & $X_{201}$ & $X_{301}$ & $X_{401}$ & All & \multicolumn{6}{c}{} \\
1.00  & 1.00      & 1.00      & 1.00      & 0.90      & 0.90 & \multicolumn{6}{c}{} \\ \cline{1-6}
\end{tabular}
}
\end{table}

\subsubsection{Simulation Study 2}
In this simulation, we generate data following the framework described in Simulation Design II. The sample size is set to $n = 200$, with $p = 500$ features and $q = 5$ outcome components. The set of true features is specified as $\mathcal{D}=\{X_1, X_{101}, X_{201}, X_{301}\}$, while the remaining $496$ features are treated as noise. Simulation performance is evaluated across 100 replicated simulation datasets.

\begin{figure}[h!]
    \centering
    \includegraphics[scale=0.2]{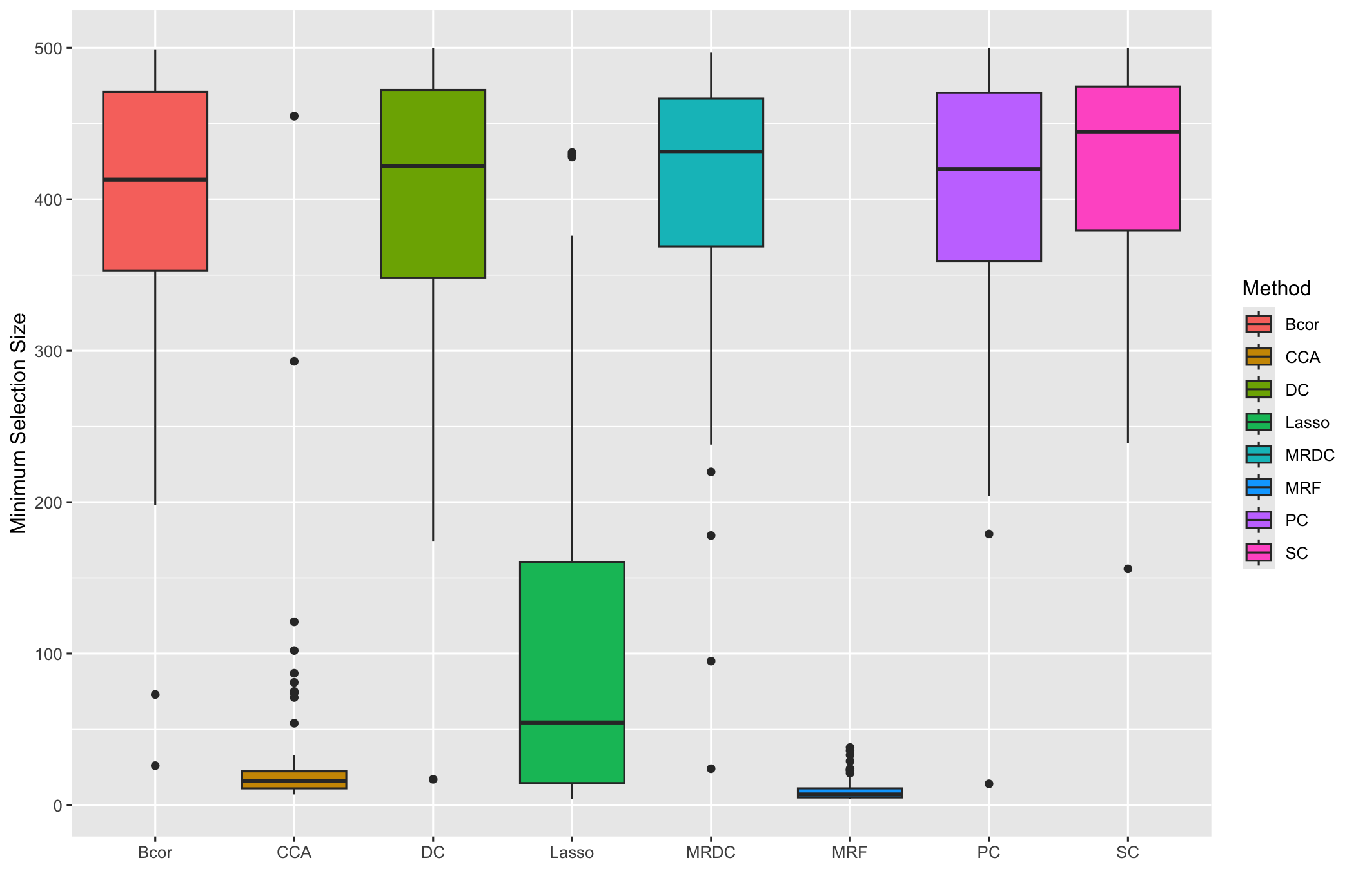}
    \vspace{-1em} 
    \caption{Boxplots of the minimum selection size $S$ in Simulation Study 2 across 100 replicated simulation datasets. The eight colors represent the eight approaches. }
    \label{fig:Fig.6}
\end{figure}

\begin{table}[h!]
\centering
\caption{
The minimum selection size $S$ for Simulation Study 2. The first row shows the mean of $S$, and the second row shows the standard error of $S$ across 100 replicated simulation datasets.
}
\label{tab:Table4}
\begin{tabular}{lcccccccc}
\hline
         & Bcor-SIS & CCA   & DC-SIS & LASSO  & MrDc-SIS & PC-Screen & SC-SIS & MRF \\ \hline
Mean of $S$ & 394.27   & 28.67 & 395.65 & 107.37 & 403.77   & 399.05    & 415.55 & \textbf{9.66} \\
SE of $S$  & 93.16    & 54.68 & 89.65  & 119.56 & 89.10    & 89.40     & 76.53  & \textbf{6.80} \\ \hline
\end{tabular}
\end{table}

\begin{table}[h!]
\caption{
Individual success rate ($P_s$) and overall success rate ($P_a$) for all eight methods in the Simulation Study 2.}
\label{tab:Table5}
\centering
\scalebox{0.9}{
\begin{tabular}{|cccc|c|cccc|c|}
\hline
\multicolumn{5}{|c|}{\textbf{Bcor-SIS}} & \multicolumn{5}{c|}{\textbf{CCA}} \\ \hline
\multicolumn{4}{|c|}{$P_s$} & $P_a$ & \multicolumn{4}{c|}{$P_s$} & $P_a$ \\ \hline
$X_1$ & $X_{101}$ & $X_{201}$ & $X_{301}$ & All &
$X_1$ & $X_{101}$ & $X_{201}$ & $X_{301}$ & All \\
0.01  & 0.02      & 0.02      & 0.04      & 0.00 &
1.00  & 1.00      & 0.89      & 0.51      & 0.51 \\ \hline

\multicolumn{5}{|c|}{\textbf{DC-SIS}} & \multicolumn{5}{c|}{\textbf{LASSO}} \\ \hline
\multicolumn{4}{|c|}{$P_s$} & $P_a$ & \multicolumn{4}{c|}{$P_s$} & $P_a$ \\ \hline
$X_1$ & $X_{101}$ & $X_{201}$ & $X_{301}$ & All &
$X_1$ & $X_{101}$ & $X_{201}$ & $X_{301}$ & All \\
0.02  & 0.03      & 0.02      & 0.02      & 0.01 &
1.00  & 1.00      & 1.00      & 0.95      & 0.95 \\ \hline

\multicolumn{5}{|c|}{\textbf{MrDc-SIS}} & \multicolumn{5}{c|}{\textbf{PC-Screen}} \\ \hline
\multicolumn{4}{|c|}{$P_s$} & $P_a$ & \multicolumn{4}{c|}{$P_s$} & $P_a$ \\ \hline
$X_1$ & $X_{101}$ & $X_{201}$ & $X_{301}$ & All &
$X_1$ & $X_{101}$ & $X_{201}$ & $X_{301}$ & All \\
0.03  & 0.02      & 0.02      & 0.02      & 0.00 &
0.03  & 0.02      & 0.02      & 0.03      & 0.01 \\ \hline

\multicolumn{5}{|c|}{\textbf{SC-SIS}} & \multicolumn{5}{c|}{\textbf{MRF}} \\ \hline
\multicolumn{4}{|c|}{$P_s$} & $P_a$ & \multicolumn{4}{c|}{$P_s$} & $P_a$ \\ \hline
$X_1$ & $X_{101}$ & $X_{201}$ & $X_{301}$ & All &
$X_1$ & $X_{101}$ & $X_{201}$ & $X_{301}$ & All \\
0.00  & 0.01      & 0.02      & 0.03      & 0.00 &
1.00  & 1.00      & 0.97      & 0.74      & \textbf{0.74} \\ \cline{1-10}
\end{tabular}
}
\end{table}

\begin{figure}[h!]
    \centering
    \includegraphics[scale=0.2]{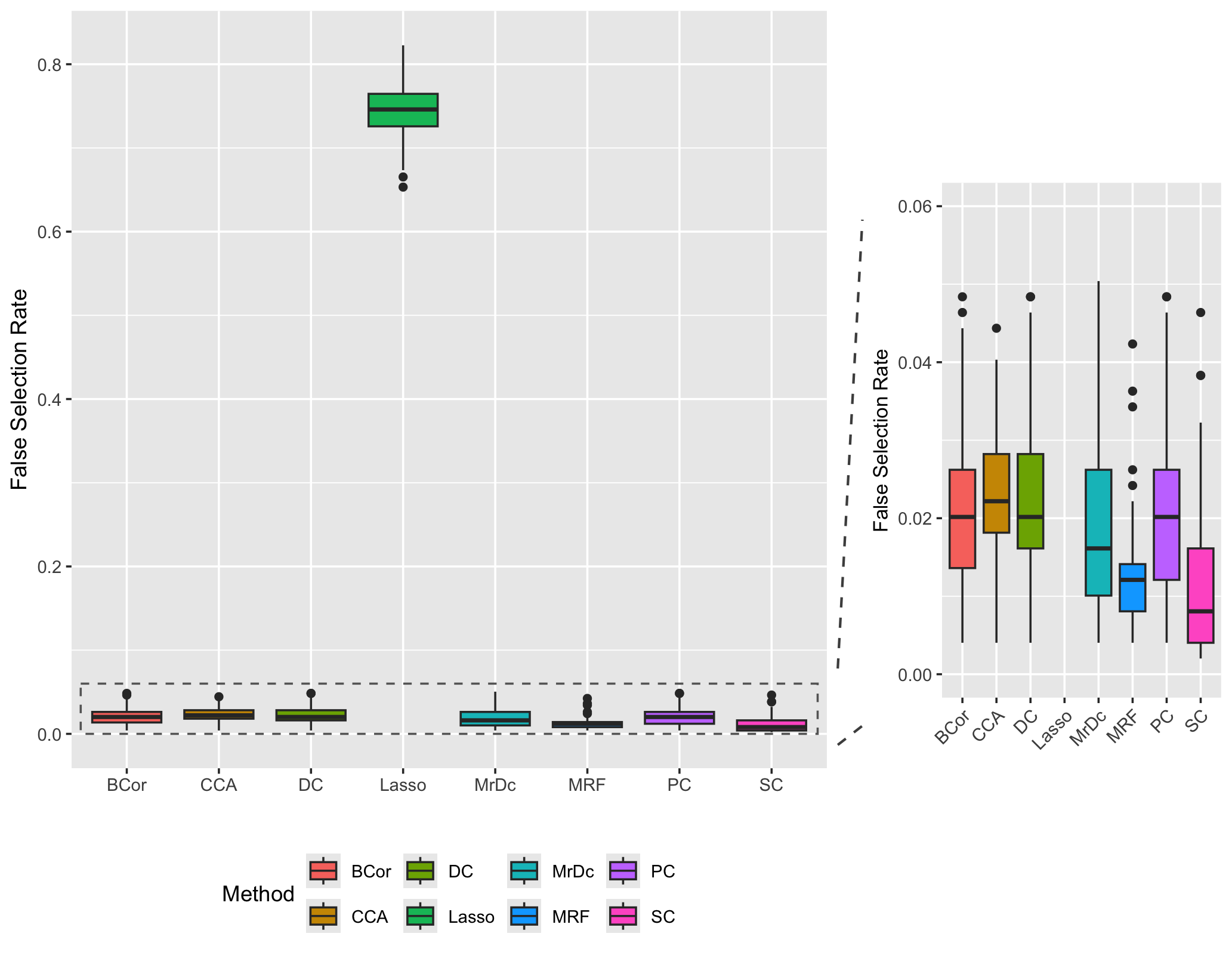}
    \vspace{-1em} 
    \caption{Boxplots of the false selection rate $T$ for Simulation Study 2 across 100 replicated datasets. Left panel: Results for all eight approaches (each color corresponds to one method). Right panel: Zoomed-in results for the seven approaches excluding LASSO.}
    \label{fig:Fig.8}
\end{figure}

Figures~\ref{fig:Fig.6} and \ref{fig:Fig.8}, together with Tables~\ref{tab:Table4} and \ref{tab:Table5}, summarize the performance of the eight approaches in Simulation Study 2. Specifically, Figure~\ref{fig:Fig.6} and Table~\ref{tab:Table4} report the minimum number of features each approach must select in order to recover all four true features. The five SIS approaches perform poorly, with average minimum selection sizes exceeding 390, implying that at least 380 additional noise features must be selected on average to match the selection accuracy of MRF. By contrast, MRF successfully recovers all four true features with an average rank of 9.66 out of 500 features across 100 replications, while CCA is the next best method, with an average minimum selection size of 28.67. Table~\ref{tab:Table5} further shows that CCA achieves an overall success rate $P_a$ of 51\%, whereas MRF reaches a higher $P_a$ of 74\%, with all individual success rates $P_s$ exceeding those of CCA. As expected, LASSO attains the highest overall success rate ($P_a = 95\%$), but this performance is offset by an inflated false selection rate, averaging 74\%. The boxplots of false selection rates ($T$) in Figure~\ref{fig:Fig.8} illustrate that MRF and SC-SIS achieve the lowest false selection rates, with MRF also exhibiting the narrowest interquartile range, underscoring its efficiency in avoiding false selections.

\subsubsection{Simulation Study 3}
The superior performance of MRF in feature selection has already been demonstrated in the first two simulation studies. However, when applying MRF directly to the real-world GWAS human facial morphology dataset with 453,273 SNPs in categorical data type, MRF breaks down. It is under expectation given that it is a joint approach that simultaneously considers all features in the same model. This inherently makes MRF more computationally demanding than independence screening methods that evaluate each feature individually, especially when the number of features $p$ is extremely large. To address this challenge and better prepare for real-data analysis, we explore an additional simulation study that is closer to the real data. In this Simulation Study 3, we build upon Simulation Design II with human facial morphology images as the outcome and multiple layers of heterogeneity, complexity, and randomness. Instead, we increase the feature dimensionality to $p = 100{,}000$ to reflect the high-dimensional nature of the real data. All features are discretized into four categories: 0 (aa), 1 (Aa or aA), 2 (AA), and 3 (missing). Due to the substantial computational burden, this simulation study is evaluated using only 10 replicated datasets.

To make MRF feasible in ultrahigh-dimensional settings, we incorporate a pre-screening step. LASSO is a natural choice given its consistently high overall success rates—98\% in Table \ref{tab:Table2}, 90\% in Table \ref{tab:Table3}, and 95\% in Table \ref{tab:Table5}—indicating that it rarely misses true features. Although LASSO is prone to high false selection rates, this characteristic is advantageous in the context of pre-screening: \textcolor{black}{by keeping a much larger set of candidates, it increases the chance that true features are retained, leaving MRF to distinguish signal from noise and refine the selection.} 

We evaluate the effectiveness of this pre-screening strategy under two scenarios:
\begin{itemize}
\item One-step scenario: MRF and CCA are each applied directly to the full feature set. With $p = 100{,}000$ features and only four true signals embedded among an overwhelming number of noise features, this setting evaluates whether MRF can still be feasible and able to discover the true features under such a challenging high-dimensional scenario.

\item Two-step scenario: LASSO is first applied to reduce the feature space from $100{,}000$ to $721$ candidates on average, after which MRF and CCA are both performed on the reduced set. This setting examines whether pre-screening improves the performance of MRF relative to its one-step scenario.
\end{itemize}

\begin{figure}[h!]
    \centering
    \includegraphics[scale=0.2]{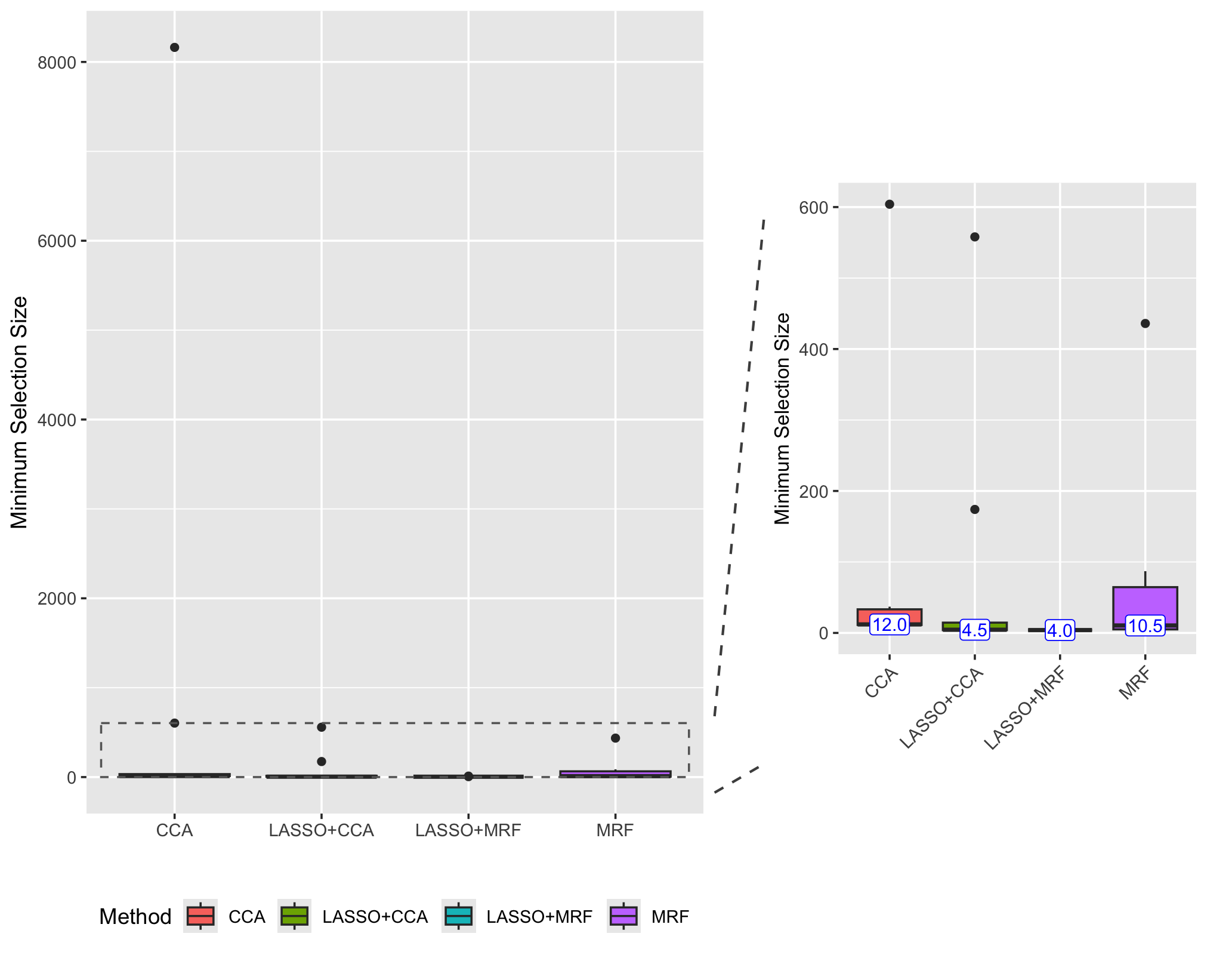}
    \vspace{-1em} 
    \caption{Boxplots of the minimum selection size $S$ in Simulation Study 3 across 10 replicated simulation datasets. Left panel: Results for the two scenarios (each color corresponds to one method). Right panel: Zoomed-in results for the two scenarios, excluding the most extreme outlier value of CCA. }
    \label{fig:Fig.9}
\end{figure}

\begin{table}[h!]
\centering
\caption{The minimum selection size $S$ for the Simulation Study 3. The first row shows the mean of $S$, and the second row shows the standard error of $S$ across 10 replicated simulation datasets.}
\label{tab:Table6}
\begin{tabular}{ccccc}
\hline
       & CCA & LASSO + CCA   & LASSO + MRF & MRF \\ \hline
Mean of $S$ & 888.1   & 78.1  & {\bf4.8} & 67.9\\
SE of $S$  & 2562.9  & 176.7  & {\bf 2.2} & 132.9\\ \hline
\end{tabular}
\end{table}

\begin{table}[h!]
\centering
\caption{
Individual success rate ($P_s$) and overall success rate ($P_a$) for four methods in the Simulation Study 3.
}
\label{tab:Table7}
\begin{tabular}{|cccc|c|cccc|c|}
\hline
\multicolumn{5}{|c|}{\textbf{CCA}} & \multicolumn{5}{c|}{\textbf{MRF }} \\ \hline
\multicolumn{4}{|c|}{$P_s$} & $P_a$ & \multicolumn{4}{c|}{$P_s$} & $P_a$ \\ \hline
$X_1$ & $X_{101}$ & $X_{201}$ & $X_{301}$ & All &
$X_1$ & $X_{101}$ & $X_{201}$ & $X_{301}$ & All \\
1.00  & 1.00      & 1.00      & 0.20      & 0.20 &
1.00  & 1.00      & 1.00      & 0.40      & 0.40 \\ \hline
\multicolumn{5}{|c|}{\textbf{LASSO + CCA}} & \multicolumn{5}{c|}{\textbf{LASSO + MRF }} \\ \hline
\multicolumn{4}{|c|}{$P_s$} & $P_a$ & \multicolumn{4}{c|}{$P_s$} & $P_a$ \\ \hline
$X_1$ & $X_{101}$ & $X_{201}$ & $X_{301}$ & All &
$X_1$ & $X_{101}$ & $X_{201}$ & $X_{301}$ & All \\
1.00  & 1.00      & 1.00      & 0.70      & 0.70 &
1.00  & 1.00      & 1.00      & 0.90      & \textbf{0.90} \\ \hline
\end{tabular}
\end{table}

\begin{figure}[h!]
    \centering
    \includegraphics[scale=0.18]{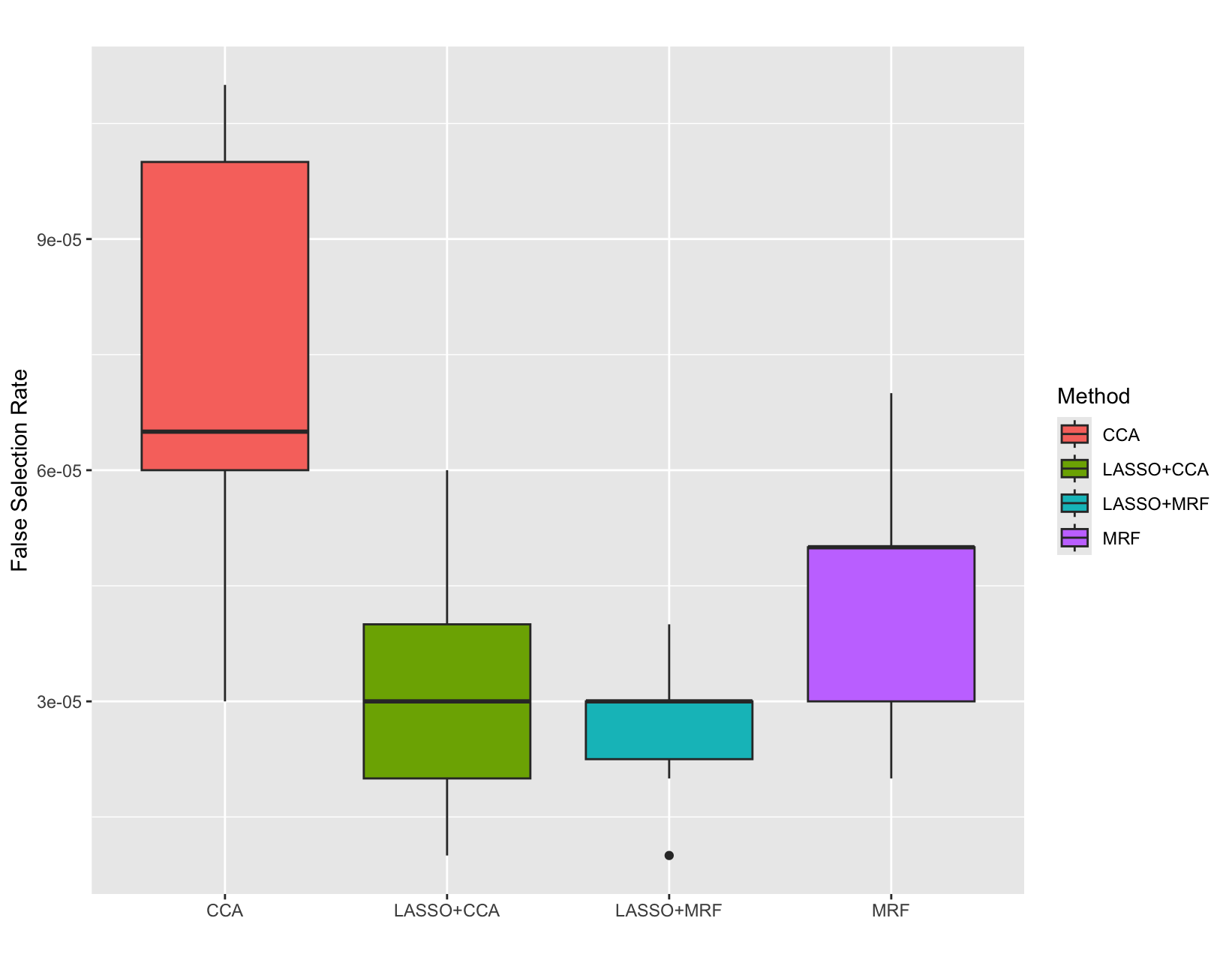}
    \caption{Boxplots of the false selection rate $T$ for Simulation Study 3 across 10 replicated simulation datasets.}
    \label{fig:Fig.10}
\end{figure}

Figures~\ref{fig:Fig.9} and \ref{fig:Fig.10}, together with Tables~\ref{tab:Table6} and \ref{tab:Table7}, summarize the results of the one-step and two-step scenarios in the Simulation Study 3. Figure~\ref{fig:Fig.9} and Table~\ref{tab:Table6} demonstrate a clear and substantial advantage of LASSO+MRF over the other three methods. This approach achieves an average minimum selection size of only 4.8 (standard error 2.2) to discover all four true features, just 7.1\% (= 4.8/67.9) of the size required by MRF without pre-screening. Somewhat unexpectedly, CCA suffers a sharp performance decline in this setting, likely due to the inherent difficulties of ultrahigh-dimensional categorical data type.

We highlight four findings from this comparison. (1) The two-step approaches achieve markedly smaller average minimum selection sizes than their one-step counterparts (4.8 vs.~67.9 for MRF; 78.1 vs.~888.1 for CCA). (2) The \textcolor{black}{standard errors} of minimum selection sizes are also dramatically reduced with pre-screening (2.2 vs.~132.9 for MRF; 176.7 vs.~2562.9 for CCA). (3) MRF consistently requires smaller selection sizes than CCA (67.9 vs.~888.1 without pre-screening; 4.8 vs.~78.1 with pre-screening). (4) MRF also attains substantially lower standard errors than those of CCA (132.9 vs.~2562.9 without pre-screening; 2.2 vs.~176.7 with pre-screening). In addition, the LASSO pre-screening step improves success rates for both methods. Specifically, the overall success rate $P_a$ of CCA increases from 20\% to 70\%, while that of MRF rises from 40\% to 90\%. Figure~\ref{fig:Fig.10} provides additional evidence on false selection rates: LASSO+MRF exhibits the narrowest interquartile range, underscoring its efficiency and accuracy in avoiding false selections.

In summary, these experiments underscore the practical advantages of incorporating LASSO as a pre-screening step in ultrahigh-dimensional applications, supporting its utility in our real data analysis: it effectively reduces the dimensionality of the candidate pool while preserving as many truly \textcolor{black}{influential} features as possible, thereby mitigating the curse of dimensionality and enhancing the effectiveness and power of subsequent joint selection by MRF. 

\section{Real Data Analysis}
The real dataset consists of high-resolution 3D facial images from 2,342 human subjects (890 males and 1,452 females). Each facial morphology is represented as a 3D point cloud comprising 7,160 quasi-landmarks, with XYZ coordinates recorded for each point \citep{Claes2012, Claes2014}. We apply alignment procedures, including translation, scaling, and rotation, to all images, a pre-processing step essential for minimizing pose-related variation and thereby revealing variation attributable to genetic factors.

 Following \citet{Claes2018}, we apply PCA to reduce the dimensionality of the original high-resolution images and only the top 50 PC scores are retained for each subset, capturing approximately $93\%$ of the total variation of facial morphologies in both male and female groups. These PC scores are utilized as the random outcome vector for our selection process. The feature dataset poses a substantial challenge due to its ultrahigh dimensionality. The original candidate pool consists of $9,478,608$ SNPs after imputation. Following standard procedures in the GWAS literature, we first apply a series of quality-control filters \citep{marees2018tutorial}. Specifically, SNPs with imputation quality below an INFO score of $0.5$ are excluded. Next, a genotype-level filter is imposed, removing genotype calls with posterior probability less than $0.9$. Monomorphic SNPs are also discarded. To further reduce redundancy and address multicollinearity caused by strong correlations among neighboring variants, linkage disequilibrium (LD) pruning is conducted by retaining only one SNP from each highly correlated block with $r^2 > 0.8$. In addition, we conduct a comprehensive literature search to ensure that all SNPs previously reported to be associated with human facial morphology are included in the initial candidate pool, thereby reducing the possibility of excluding true signals prior to the feature selection stage. After applying these pre-processing steps, a number of $453,273$ features is obtained for the subsequent selection. The MRF model, as a joint model, breaks down to such ultrahigh dimensionality. Guided by the strategies from the Simulation Study 3, we implement a LASSO-based pre-screening step, which reduces the number of features to $66,979$ for the male subset and $97,786$ for the female subset.

\begin{figure}[p]
    \centering
    \includegraphics[scale=0.2]{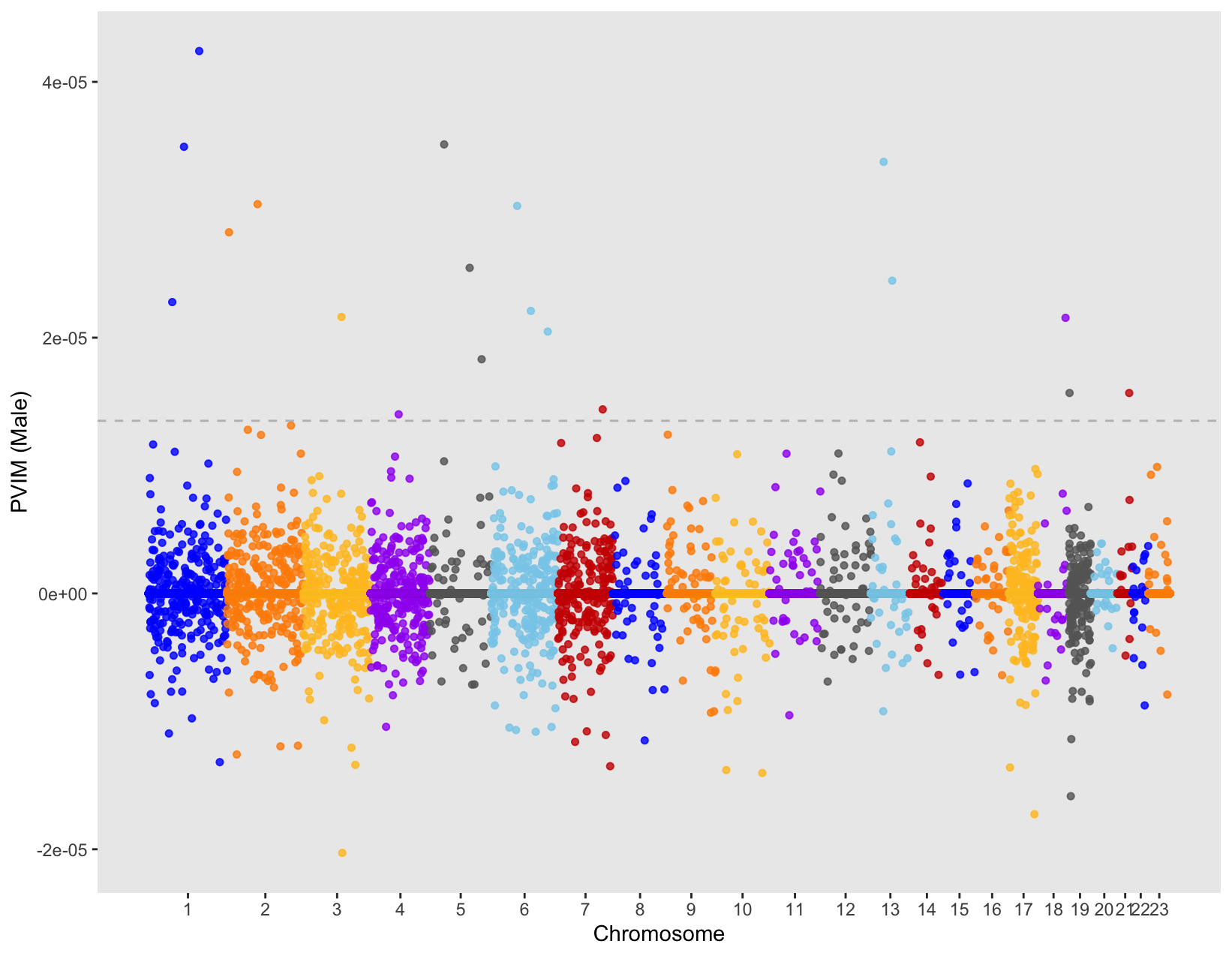}
    \caption{Permutation-based variable importance measures for all features in the male group. \textcolor{black}{SNPs excluded by LASSO during pre-processing are added back with PVIM values set to zero.} A total of 19 SNPs exceed the threshold of $1.35 \times 10^{-5}$.}
    \label{fig:Fig.11}
    \includegraphics[scale=0.2]{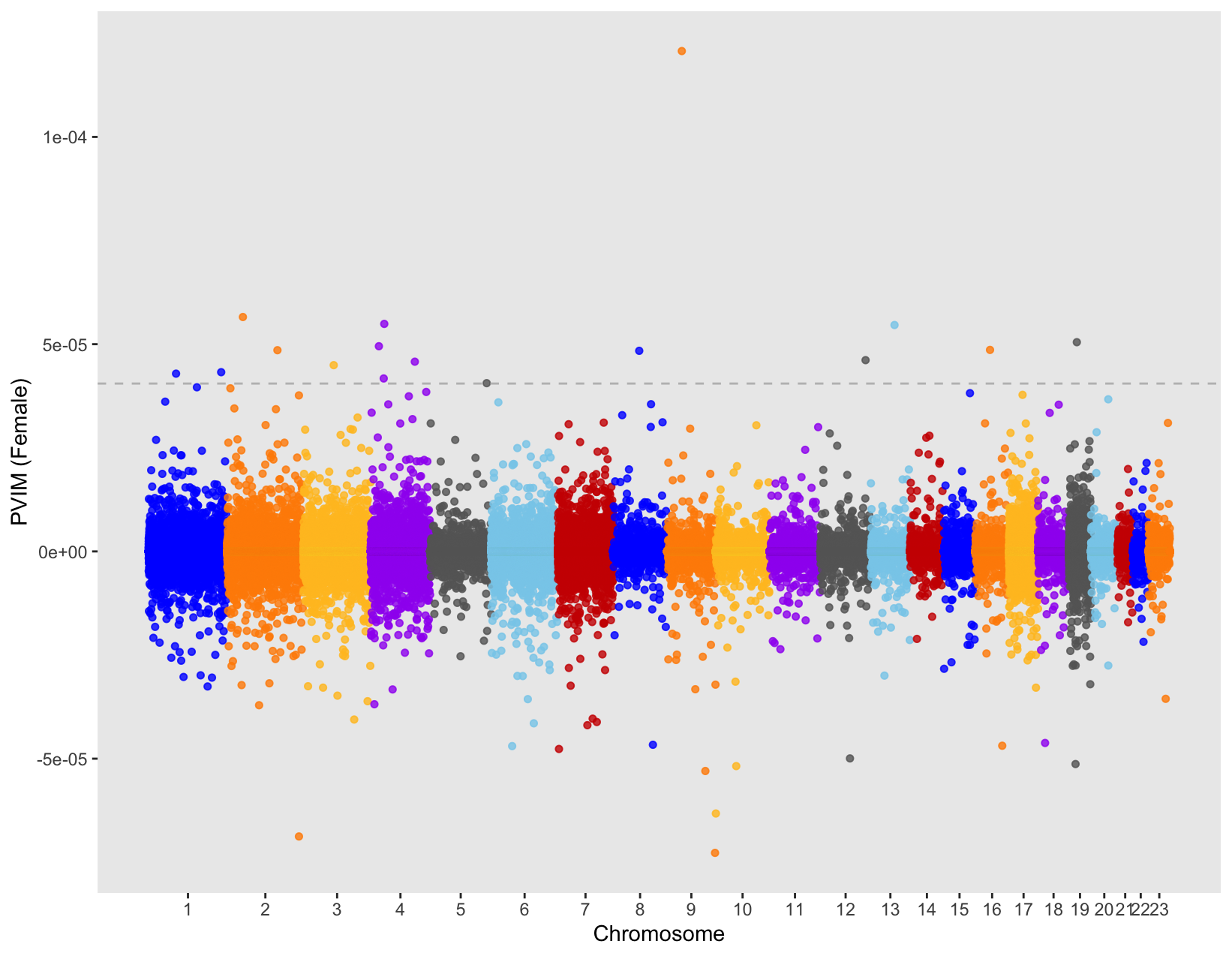}
    \caption{Permutation-based variable importance scores for all features in the female group. \textcolor{black}{SNPs excluded by LASSO during pre-processing are added back with PVIM values set to zero.} A total of 16 SNPs exceed the threshold of $4.05 \times 10^{-5}$.}
    \label{fig:Fig.12}
\end{figure}

\begin{table}[h!]
\centering
\caption{Important SNPs selected by the MRF procedure, with candidate genes identified through functional annotation. \textcolor{black}{Red (blue) highlights indicate individual genes (gene–gene interactions) supported by prior independent studies in the literature, conducted using completely different datasets, methodologies, and experimental systems, including molecular studies.}}
\label{tab:Table8}
\begin{tabular}{ccc|ccc}
\hline
\multicolumn{3}{c|}{Male}                & \multicolumn{3}{c}{Female}               \\ \hline
Ch. & SNP        & Candidate gene & Ch. & SNP        & Candidate gene \\ \hline
1  & rs6608458   & \textcolor{red}{MROH9}         & 1  & rs505867   & PIGK          \\
1  & rs17648010  & SASS6          & 1  & rs283694   & \textcolor{blue}{RXRG}          \\
1  & rs4423051   & CACHD1         & 2  & rs4670733  & LINC00211     \\
2  & rs10931818  & \textcolor{red}{SATB2}          & 2  & rs10169306 & CACNB4        \\
2  & rs4252008   & \textcolor{blue}{ACTR3}          & 3  & rs6770667  & PROK2         \\
3  & rs13061458  & GPD1L          & 4  & rs1714058  & ARAP2         \\
4  & rs55866514  & FRAS1          & 4  & rs207305   & ERCC5         \\
5  & rs6872795   & SPEF2          & 4  & rs13137679 & DCLK2         \\
5  & rs1990962   & PRDM6          & 4  & rs17029777 & \textcolor{red}{DDIT4L}        \\
5  & rs2546375   & LOC285627      & 5  & rs702098   & RANBP17       \\
6  & rs13194549  & GPR111         & 8  & rs1871217  & CLVS1         \\
6  & rs211215    & FBXL4          & 9  & rs1412378  & TUSC1         \\
6  & rs661005    & \textcolor{blue}{PLEKHG1}        & 12 & rs10773092 & \textcolor{blue}{NCOR2}         \\
7  & rs10954341  & PODXL          & 13 & rs7332088  & LINC00331     \\
13 & rs842405    & HTR2A          & 16 & rs185375   & TOX3          \\
13 & rs7981863   & KLF12          & 19 & rs438421   & IFI30         \\
18 & rs12456915  & OTX2           &    &            &               \\
19 & exm1398358  & GRIN3B         &    &            &               \\
21 & rs2837768   & DSCAM          &    &            &               \\ \hline
\end{tabular}
\end{table}

After computing PVIM scores from the MRF, thresholds are determined following the same rule used in the simulation studies. That says, features are retained only if they rank within the top 5\% by PVIM and simultaneously among the top five max-ratio thresholds. This procedure yields thresholds of $0.000013$ for the male group and $0.00004$ for the female group. Figures \ref{fig:Fig.11} and \ref{fig:Fig.12} display the PVIM distributions for all features in the male and female groups, respectively. In total, 19 SNPs in the male group and 16 SNPs in the female group surpass their respective thresholds. Table~\ref{tab:Table8} lists the SNPs identified by the proposed MRF procedure. The corresponding candidate genes are annotated using the NCBI dbSNP database%
\footnote{\url{https://www.ncbi.nlm.nih.gov/snp}} 
and the UCSC Genome Browser%
\footnote{\url{https://www.genome.ucsc.edu/}}. 
For SNPs not located within any annotated gene region, 
the nearest protein-coding gene is reported as the candidate gene. Most of our findings are novel, \textcolor{black}{with some supported by prior studies.} Specifically, three genes, $SATB2$, $MROH9$, and $DDIT4L$, have been independently identified in the literature (highlighted in red), while two gene–gene interactions, $ACTR3 \times PLEKHG1$ and $RXRG \times NCOR2$, have also been previously documented (highlighted in blue). Variants in $SATB2$ are known to influence midface height, jaw structure, nose morphology, and craniofacial dysmorphology \citep{Pickrell2016, Rainger2014, Conte2016, SheehanRooney2010}. Variants in $MROH9$ have been linked to traits such as nose width, cheekbone structure, mouth morphology, and jaw morphology \citep{White2021}, while $DDIT4L$ has been associated with nose size \citep{Pickrell2016}. Notably, the two interactions detected by our model are supported by evidence from molecular studies. The interaction $ACTR3 \times PLEKHG1$ was identified by \citet{Go2021} using a proximity-based proteomic mapping approach (BioID) that captures spatially co-localized proteins. In contrast, the $RXRG \times NCOR2$ interaction was identified by \citet{Albers2005} using yeast two-hybrid screening, indicating direct physical binding between the proteins.

\begin{table}[h!]
\caption{Important two-way gene–gene interactions identified by MRF using the interactive-PVIM score. In the male subset, the top three interactions all involve $GPD1L$, while in the female subset, the top four interactions all involve $GNG4$.}
\label{tab:Table9}
\scalebox{0.72}{
\begin{tabular}{lc|lc}
\hline
\multicolumn{2}{c|}{Male} & \multicolumn{2}{c}{Female} \\ \hline
\multicolumn{1}{c}{Interaction} & Interactive-PVIM  & \multicolumn{1}{c}{Interaction} & Interactive-PVIM \\ \hline
rs13061458 × rs6608458 (GPD1L × MROH9)  & 0.04290 & rs10926189 × rs1412378 (GNG4 × TUSC1)   & 0.04655 \\
rs13061458 × rs6872795 (GPD1L × SPEF2)  & 0.04289 & rs10926189 × rs13137679 (GNG4 × DCLK2)  & 0.04274 \\
rs13061458 × rs17648010 (GPD1L × SASS6) & 0.04289 & rs10926189 × rs10773092 (GNG4 × NCOR2)  & 0.04256 \\
                                       &          & rs10926189 × rs6770667 (GNG4 × PROK2)   & 0.04195 \\ \hline
\end{tabular}}
\end{table}

Although tree structures naturally capture interaction effects among features along each branch, it is difficult to directly interpret such interactions for MRF because the final prediction aggregates across a large number of trees, each containing several branches. Consequently, when a feature attains a high PVIM score, it is often unclear whether this reflects a strong main effect or relatively weak main effects combined with strong interactions with other features. Thus, a high PVIM score in MRF may indicate either scenario. To better evaluate two-way interactions, we compute interactive-PVIM scores using the R package \texttt{randomForestSRC} \textcolor{black}{for $66{,}979$ SNPs in the male subset and $97{,}786$ SNPs in the female subset.} For a given feature pair, this score is obtained by permuting both features jointly as a unit and calculate their joint-PVIM. The interactive-PVIM is then defined as the absolute difference between the joint-PVIM for the pair and the sum of their individual PVIM scores. Table~\ref{tab:Table9} reports the top seven pairs of two-way interactions identified by MRF based on this interactive-PVIM measure. Notably, in both the male and female groups, the strongest interactive-PVIM signals involve one gene interacting with multiple partners. In males, the top three pairs all involve $GPD1L$, whereas in females, the top four pairs all involve $GNG4$. To the best of our knowledge, the findings in Table~\ref{tab:Table9} are novel. They illustrate the potential of MRF to uncover hub-gene interactions, i.e., an important biological phenomenon in which a single gene engages in multiple gene–gene interactions. This can reveal new insights into the genetic architecture of biomedical traits. 

\section{Discussion}

 Recent literature has made notable progress on univariate random forests (URF). In particular, consistency of URF predictions has been established \citep{Biau2012, Scornet2015}. \citet{Mentch2016} derived formal statistical inference, including confidence intervals and hypothesis testing, for both predictions and variable importance in URF. While these advances have substantially deepened understanding of URF, comparable theoretical guarantees for MRF, particularly in the context of feature selection, remain scarce.

 {\color{black}
In addition to the permutation-based variable importance measure, a few alternative importance metrics for MRF have been proposed. For example, naive importance methods assess how frequently a feature is used in splitting, while split-based approaches aggregate reductions in prediction error on OOB samples attributable to each feature at the splits \citep{sikdar2025}. Although these methods provide a structure-driven evaluation of feature importance, they primarily capture a feature’s contribution to local split improvements rather than its impact on overall predictive performance. Moreover, these approaches currently lack theoretical guarantees.}

In this article, we address two key challenges that enable multivariate random forests to serve as an effective and accurate feature selection framework. \textcolor{black}{From a theoretical perspective, we establish the first consistency result for feature selection with multivariate outcomes using MRF in high-dimensional settings}. This provides a theoretical guarantee that MRF will not miss truly \textcolor{black}{influential} features when sample size is sufficiently large. The MRF integrates feature selection directly into a nonparametric, distribution-free learning process and identifies \textcolor{black}{influential} features based on their contributions to predictive accuracy.

From a practical perspective, we experiment a two-step LASSO–MRF strategy to overcome the computational breakdown of MRF in ultrahigh-dimensional settings and to further enhance its efficiency and accuracy. We apply MRF to a genome-wide association study of human facial morphology, a complex trait influenced by a large number of genetic factors and their interactions. Traditional GWAS methods analyze each SNP in isolation and therefore fail to detect SNPs with weak main effects but strong joint effects through interactions. Among independence feature selection approaches, canonical correlation analysis, a widely used GWAS method, was applied to this same dataset in a previous study \citep{Claes2018}. In contrast, the PVIM of MRF simultaneously captures complex dependence structures, including multicollinearity, nonlinear and higher-order interactions among SNPs, and heterogeneity across subjects, by aggregating information across an ensemble of trees. \textcolor{black}{Some of our findings are supported by prior molecular studies, while the majority are novel. Our findings may motivate further molecular and genetic investigations into these detected gene interactions and their role in shaping human facial morphology.}

Several practical challenges remain. First, PVIM is more computationally intensive than independence screening methods; therefore, scaling to massive datasets will require further algorithmic and computational optimization. Second, we currently adopt a max-ratio thresholding rule, which may not meet all application needs. Therefore, developing theoretically grounded procedures for statistical inference, such as p-values or confidence intervals for feature importance, remains an important direction for future work.



\newpage

\appendix
\section{}
\label{app:theorem}

\begin{proof}[Proof of Theorem]
The first goal is to show that $\underset{1 \leq j \leq p}{\max}\mathrm{Var}(\hat{\Lambda}_n^{(j)}) = O(R_n)$. For each $j \in \{1,...,p\}$ we can express $\hat{\Lambda}_n^{(j)}$ as the difference 
$$\hat{\Lambda}^{(j)}_n = \Gamma_{nj} - \Gamma_{n}$$where$$\Gamma_{nj} = \frac{1}{m_n(n-k_n)} \sum_{m=1}^{m_n}\sum_{i =1}^{n} g_{nj}(\mathbf{Z}_i, \mathcal{Z}_{S_m}, \Omega_{mj}) \mathbb{1}_{\{i\notin S_m\}}  $$and 
\[
\Gamma_n = \frac{1}{m_n(n-k_n)} \sum_{m=1}^{m_n}\sum_{i =1}^{n} g_{n}(\mathbf{Z}_i, \mathcal{Z}_{S_m}, \omega_m)\mathbb{1}_{\{i\notin S_m\}} .
\]
Thus we have \[
\underset{1 \leq j \leq p}{\max}\mathrm{Var}(\hat{\Lambda}_n^{(j)}) \leq 2\Bigl(\underset{1 \leq j \leq p}{\max}\mathrm{Var}(\Gamma_{nj}) + \mathrm{Var}(\Gamma_n)\Bigr).
\]
Hence is suffices to prove that,
\[
(i) \, \underset{1 \leq j \leq p}{\max}\mathrm{Var}(\Gamma_{nj}) = O(R_n)\quad \text{and} \quad (ii) \, \mathrm{Var}(\Gamma_{n}) = O(R_n) .
\]
We first establish the result in part $(i)$.
For  $m = 1, \dots, m_n$, $j =1,...,p$, and $i \notin S_m$ define 
\begin{align*}
\tilde{g}_{nj}(\mathbf{Z}_i, \mathcal{Z}_{S_m}, \Omega_{mj}) &= g_{nj}(\mathbf{Z}_i, \mathcal{Z}_{S_m}, \Omega_{mj}) - \bar{g}_{nj}(\mathbf{Z}_i, \mathcal{Z}_{S_m}), \\
\tilde{e}_{nj}(\mathbf{Z}_i, \mathcal{Z}_{S_m}) &= \bar{g}_{nj}(\mathbf{Z}_i, \mathcal{Z}_{S_m}) - h_{nj}(\mathcal{Z}_{S_m})
    \quad \text{where}, \\
     h_{nj}(\mathcal{Z}_{S_m}) &= \mathbf{E}\Big[\bar{g}_{nj}(\mathbf{Z}_i, \mathcal{Z}_{S_m}) \, | \, \mathbf{Z}_1,...,\mathbf{Z}_{i-1}, \mathbf{Z}_{i+1},...,\mathbf{Z}_n, \mathcal S_{k_n,m_n} \Big].
\end{align*}
Then $\Gamma_{nj}$ can be expressed as a sum:
\begin{equation}
    \Gamma_{nj} = A_{nj} + B_{nj} + U_{nj}
\end{equation}
where
\begin{align*}
    A_{nj} &= \frac{1}{m_n(n-k_n)}\sum_{m=1}^{m_n}\sum_{i=1}^{n} \tilde{g}_{nj}\bigl(\mathbf{Z}_i, \mathcal{Z}_{S_m}, \Omega_{mj}\bigr)\mathbb{1}_{\{i\notin S_m\}}, \\
    B_{nj} &= \frac{1}{m_n(n-k_n)}\sum_{m=1}^{m_n}\sum_{i=1}^{n} \tilde{e}_{nj}\bigl(\mathbf{Z}_i, \mathcal{Z}_{S_m})\mathbb{1}_{\{i\notin S_m\}} ,  \\
    U_{nj} &= \frac{1}{m_n(n-k_n)}\sum_{m=1}^{m_n}\sum_{i=1}^{n} h_{nj}\bigl(\mathcal{Z}_{S_m}) \mathbb{1}_{\{i\notin S_m\}} .
\end{align*}
Now since $A_{nj}$ and $B_{nj}$ are centered we have
\[
 \underset{1 \leq j \leq p}{\max}\mathrm{Var}(\Gamma_{nj}) \leq  3\Bigl( \underset{1 \leq j \leq p}{\max}\mathbf{E}(A_{nj}^2) +  \underset{1 \leq j \leq p}{\max}\mathbf{E}(B_{nj}^2) +  \underset{1 \leq j \leq p}{\max}\mathrm{Var}(U_{nj}) \Bigr).
\]
Thus, we derive suitable asymptotic bounds for the second moments of $A_{nj}, B_{nj}$ and $U_{nj}$.  Specifically, we proceed in three steps and obtain the following results:
\begin{center}
    \begin{minipage}{0.4\textwidth} 
        \begin{itemize}
            \item[(1)] $ \underset{1 \leq j \leq p}{\max}\mathbf{E}(A_{nj}^2) = O\bigl(\frac{n}{m_n}\bigr)$,
            \item[(2)] $ \underset{1 \leq j \leq p}{\max}\mathbf{E}(B^2_{nj}) = O\bigl(\frac{k_n^{\alpha }}{n}\bigr)$,
            \item[(3)] $ \underset{1 \leq j \leq p}{\max}\mathrm{Var}(U_{nj}) = O\bigl(\frac{k_n^{2+\alpha}}{n}\bigr)$.
        \end{itemize}
    \end{minipage}
\end{center}
\noindent    

\noindent
The following notation is used repeatedly in the proof. For a random variable $W$, define $$\mathbf{E}_{\mathcal{S}}(W) := \mathbf{E}[W \, | \, \mathcal S_{k_n,m_n}] \quad \text{and} \quad \mathbf{E}_{ \mathbf{Z}, \mathcal{S}}(W) = \mathbf{E}[W \, | \, \mathbf{Z}_1,...,\mathbf{Z}_n, \mathcal S_{k_n,m_n}].$$

\noindent
\textbf{{Proof of (1).}}
For $j =1,...,p$, write
\begin{align*}
  A_{nj} &= \frac{1}{m_n}\sum_{m =1}^{m_n} \rho_{mj} \quad \text{with,} \quad
  \rho_{mj} = \frac{1}{n-k_n}\sum_{i =1}^{n} \tilde{g}_{nj}(\mathbf{Z}_i, \mathcal{Z}_{S_m}, \Omega_{mj}) \mathbb{1}_{\{i\notin S_m\}}.
\end{align*}
Observe that since $\Omega_1,...,\Omega_{m_n}$ are i.i.d,  the random variables \textcolor{black}{ $\rho_{1j},...,\rho_{m_nj}$ }are conditionally independent and centered given $\mathbf{Z}_1,...,\mathbf{Z}_n$ and $\mathcal S_{k_n,m_n}$. 
Hence, for $m \neq l$ with $m,l = 1,\dots,m_n$, we have
\[
\mathbf{E}_{\mathbf Z,\mathcal S}(\rho_{mj} \rho_{lj})
=
0 .
\]
This implies that $\mathbf{E}(\rho_{mj} \rho_{lj}) =0$ for all $m,l =1,..,m_n$ and $m \neq l$. Also since $\mathbf Z_1,\dots,\mathbf Z_n $ are i.i.d,
$S_1,\dots,S_{m_n} \overset{\mathrm{iid}}{\sim} \mathrm{Unif}(\mathcal A_{n,k_n})$, and $\Omega_1,\dots,\Omega_{m_n}$ are i.i.d., it follows that
$\rho_{1j},..,\rho_{m_nj}$ are identically distributed. Thus  $\rho_{1j},..,\rho_{m_nj}$ are pairwise uncorrelated, centered and identically distributed random features.  
Then for any $m =1,..,m_n$, $\mathbf{E}(\rho_{mj}^2) = \mathbf{E}(\rho_{1j}^2)$. Further by Cauchy-Schwartz inequality
\[
\mathbf{E}(\rho_{1j}^2)
\leq
\mathbf{E}\!\left[
\tilde g_{nj}^2(\mathbf Z_{k_n+1}, \mathbf Z_1, \dots, \mathbf Z_{k_n}, \Omega_1)
\right].
\]
Hence, by condition $\textit{(C1)}$,
\[
 \max_{1 \leq j \leq p} \mathbf{E}(\rho_{1j}^2)= O(n).
\]
\noindent
Therefore 
\begin{align}
   \max_{1 \leq j \leq p} \mathbf{E}\Big[A_{nj}^2\Big]  
    = \frac{1}{m_n} \max_{1 \leq j \leq p}\mathbf{E}(\rho_{1j}^2) 
    = O\Bigl(\frac{n}{m_n}\Bigr).
\end{align}


\noindent
\textbf{{Proof of (2).}}
We begin with some preliminary observations. Since $(\mathbf{Z}_1,...,\mathbf{Z}_n)$ are i.i.d., for all $i = 1,...,n$ and $m = 1,...,m_n$ with $i \notin S_m$, the following holds:\[\tilde{e}_{nj}(\mathbf{Z}_i, \mathcal{Z}_{S_m}) \overset{d}{=} \tilde{e}_{nj}(\mathbf{Z}_{k_n+1},\mathbf{Z}_1,...,\mathbf{Z}_{k_n}).
\]
Note also that $\mathbf{E}\Big[ \tilde{e}_{nj}^2(\mathbf{Z}_{k_n+1},\mathbf{Z}_1,...,\mathbf{Z}_{k_n})\Big] \le \mathbf{E}\Big[ \bar{g}_{nj}^2(\mathbf{Z}_{k_n+1},\mathbf{Z}_1,...,\mathbf{Z}_{k_n})\Big]$; thus, under condition \textit{(C2)}, we have
\[
\max_{1 \leq j \leq p}\mathbf{E}\Big[ \tilde{e}_{nj}^2(\mathbf{Z}_{k_n+1},\mathbf{Z}_1,...,\mathbf{Z}_{k_n})\Big] = O(k_n^{\alpha}).
\]
Further by definition $\mathbf{E}\Bigl(\tilde{e}_{nj}(\mathbf{Z}_i, \mathcal{Z}_{S_m})\Bigr) = 0$ for all $i=1,...,n$; $m=1,...,m_n$ and $j=1,...,p$. 
Now \textcolor{black}{we analyze the conditional expectations
\[
\mathbf{E}_{\mathcal{S}}\!\Big[
\tilde{e}_{nj}(\mathbf{Z}_{i_1}, \mathcal{Z}_{S_m})
\,\tilde{e}_{nj}(\mathbf{Z}_{i_2}, \mathcal{Z}_{S_l}) 
\Big],
\]
for $i_1,i_2 = 1,\ldots,n$ and $m,l = 1,\ldots,m_n$. The following holds by Cauchy-Schwartz inequality,
\begin{align*}
    \mathbf{E}_{\mathcal{S}}\!\Big[\tilde{e}_{nj}(\mathbf{Z}_{i_1}, \mathcal{Z}_{S_m}) \,
    \tilde{e}_{nj}(\mathbf{Z}_{i_2}, \mathcal{Z}_{S_l}) \Big]
&\leq \mathbf{E}_{\mathcal{S}}\left[\tilde{e}_{nj}^2(\mathbf{Z}_{k_n+1},\mathbf{Z}_1,...,\mathbf{Z}_{k_n})  \right] \\
&= \mathbf{E}\left[\tilde{e}_{nj}^2(\mathbf{Z}_{k_n+1},\mathbf{Z}_1,...,\mathbf{Z}_{k_n})  \right].
\end{align*}
for all  $m,l = 1,\ldots,m_n$ and for all $i_1 \notin S_m, i_2 \notin S_l$.
However most of these terms are zero as can be seen in the following cases.} 
\begin{itemize}
        \item[(a1)] If $i_2 \notin \{i_1\} \cup S_m \cup S_l $ and \textcolor{black}{$i_1 \notin S_m$} then
    \begin{align*}
     &\, \mathbf{E}_{\mathcal{S}}\!\Big[\tilde{e}_{nj}(\mathbf{Z}_{i_1}, \mathcal{Z}_{S_m}) \,
    \tilde{e}_{nj}(\mathbf{Z}_{i_2}, \mathcal{Z}_{S_l}) \Big]\\
    = & \, \mathbf{E}_{\mathcal{S}}\!\Big[
        \mathbf{E}\Big[\tilde{e}_{nj}(\mathbf{Z}_{i_1}, \mathcal{Z}_{S_m}) \,
        \tilde{e}_{nj}(\mathbf{Z}_{i_2}, \mathcal{Z}_{S_l})
        \,\big|\, \mathbf{Z}_1,..,\mathbf{Z}_{i_2-1}, \mathbf{Z}_{i_2+1},..,\mathbf{Z}_n,\mathcal S_{k_n,m_n}  \Big]
    \Big].\\
    = & \, \mathbf{E}_{\mathcal{S}}\!\Big[ \tilde{e}_{nj}(\mathbf{Z}_{i_1}, \mathcal{Z}_{S_m})
        \mathbf{E}\Big[\tilde{e}_{nj}(\mathbf{Z}_{i_2}, \mathcal{Z}_{S_l}) \big|\, \mathbf{Z}_1,..,\mathbf{Z}_{i_2-1}, \mathbf{Z}_{i_2+1},..,\mathbf{Z}_n, \mathcal S_{k_n,m_n}\Big] 
    \Big]\\
    = & \, 0.
\end{align*}
          
 \item[(a2)] If $\,i_1 \notin \{i_2\} \cup S_{m} \cup S_l $ and \textcolor{black}{$i_2 \notin S_l$} then by an argument symmetric to that in (a1), it follows that
 \[\mathbf{E}_{\mathcal{S}}\!\Big[\tilde{e}_{nj}(\mathbf{Z}_{i_1}, \mathcal{Z}_{S_m}) \,
    \tilde{e}_{nj}(\mathbf{Z}_{i_2}, \mathcal{Z}_{S_l}) \Big] = 0.\]
              
\end{itemize}
    
\noindent
Now let\[
m_n^2 (n-k_n)^2 \mathbf{E}\bigl(B_{nj}^2\bigr)  \notag 
= \mathbf{E}(b_1) + \mathbf{E}(b_2)
\]
where
\begin{align*}
b_1 &= \sum_{l=1}^{m_n} \sum_{m=1}^{m_n} \sum_{i_1 = 1}^{n}\sum_{i_2 = 1}^{n}  \mathbf{E}_{\mathcal{S}}\!\big[\tilde{e}_{nj}(\mathbf{Z}_{i_1}, \mathcal{Z}_{S_m}) \tilde{e}_{nj}(\mathbf{Z}_{i_2}, \mathcal{Z}_{S_{l}})  \big]\, \mathbb{1}_{\{ i_1 \neq i_2 \}} \mathbb{1}_{\{i_1\notin S_m\}} \mathbb{1}_{\{i_2\notin S_l\}}  \\[10pt]
b_2 &= \sum_{l=1}^{m_n} \sum_{m=1}^{m_n} 
     \sum_{i_1 = 1}^{n} \mathbf{E}_{\mathcal{S}}\!\big[\tilde{e}_{nj}(\mathbf{Z}_{i_1}, \mathcal{Z}_{S_m}) \tilde{e}_{nj}(\mathbf{Z}_{i_1}, \mathcal{Z}_{S_{l}}) \big]\, \mathbb{1}_{\{i_1\notin S_m \cup S_l\}}  .
\end{align*}
\textcolor{black}{Based} on the arguments in (a1) and (a2), $b_1$ contains at most \textcolor{black}{$k_n^2m_n^2$} non-zero elements, whereas $b_2$ contains \textcolor{black}{at most} $(n-k_n)m_n^2$ \textcolor{black}{non-zero} elements. Further the absolute value of each non-zero term in $b_1$ and $b_2$ is bounded by $\mathbf{E}\left[\tilde{e}_{nj}^2(\mathbf{Z}_{k_n+1},\mathbf{Z}_1,...,\mathbf{Z}_{k_n}) \right]$. Consequently, each term in $\mathbf{E}(b_1)$ and $\mathbf{E}(b_2)$ also admits the same bound. 
Therefore, \[
\mathbf{E}[B_{nj}^2] \leq 
 \Big(\frac{k_n^2 }{(n-k_n)^2} + \frac{1}{n-k_n} \Big)\mathbf{E}\Big[\tilde{e}^2_{nj}(\mathbf{Z}_{k_n+1}, \mathbf{Z}_1,  \ldots, \mathbf{Z}_{k_n})\Big],\]
and it follows that
\begin{equation} \label{variance_bound}
\begin{split}
\max_{1 \leq j \leq p}\mathbf{E}[B_{nj}^2]&\leq 
 \Big(\frac{k_n^2 }{(n-k_n)^2} + \frac{1}{n-k_n} \Big)\max_{1 \leq j \leq p} \mathbf{E}\Big[\tilde{e}^2_{nj}(\mathbf{Z}_{k_n+1}, \mathbf{Z}_1,  \ldots, \mathbf{Z}_{k_n})\Big] \\
 &= \Big(\frac{k_n^2 }{(n-k_n)^2} + \frac{1}{n-k_n} \Big) O(k_n^{\alpha}) \notag\\
 &= O\Bigl( \frac{k_n^ \alpha}{n} \Bigr) \notag.
\end{split}
\end{equation}
Here the last step follows from the assumption that
$\underset{n \to \infty}{\lim} \frac{k_n^{2+\alpha}}{n} = 0,$
 which implies:
\[\frac{1}{n-k_n} = O\!\left(\frac{1}{n}\right) \quad \text{and} \quad \frac{k_n^2}{n^2} = O\Bigl(\frac{1}{n}\Bigr).\]

\noindent
\textbf{Proof of (3).} Recall, $h_{nj}(\mathcal{Z}_{S_m}) = \mathbf{E}\Big[\bar{g}_{nj}(\mathbf{Z}_i, \mathcal{Z}_{S_m}) \, | \, \mathbf{Z}_1,...,\mathbf{Z}_{i-1}, \mathbf{Z}_{i+1},...,\mathbf{Z}_n, \mathcal S_{k_n,m_n} \Big]$. The function stays the same regardless of order of elements in $\mathcal{Z}_{S_m}$, thus for each $j \in \{1,..,p\}$, $h_{nj}(.)$ is permutation symmetric, and $U_{nj}$ can be thought of as an incomplete U-statistic of order $k_n$ with kernel $h_{nj}$. 

\noindent
The following notation is required to express the variance of a U-statistic. 

For $c =1,...,k_n$ and $j=1,...,p$, define  
\begin{align}
    \zeta^{(j)}_{c,k_n} 
&:= \operatorname{Var} ~\!\Bigl\{\,
\mathbf{E}\!\Big[h_{nj}(\mathbf{Z}_1, \ldots, \mathbf{Z}_{k_n}) 
\,\big|\, 
\mathbf{Z}_1, \ldots, \mathbf{Z}_c\Big]
\Bigr\}.
\end{align}
So, $\zeta^{(j)}_{c,k_n} \geq 0$ for all $c =1,..,k_n$, and the following inequalities hold: 
\[
\zeta^{(j)}_{1,k_n} \;\le\; \cdots \;\le\; \zeta^{(j)}_{k_n,k_n}.
\]
In addition, for each $j \in \{1,..,p\}$, let \(U^*_{nj}\) denote  the corresponding \emph{complete} U-statistic, obtained by averaging over all \(\binom{n}{k_n}\) subsets of size \(k_n\) \textcolor{black}{ (i.e., $\mathcal{A}_{n,k_n}$).}

The variance of an incomplete U-statistic (in our setting, denoted by $U_{nj}$)  
based on $m_n$ subsamples selected uniformly at random with replacement \textcolor{black}{drawn from $\mathcal{A}_{n,k_n}$}, 
is given by   
\begin{equation} \label{incomplete u-stat}
   \mathrm{Var}(U_{nj}) = \frac{\zeta^{(j)}_{k_n,k_n}}{m_n} + \Bigl(1- \frac{1}{m_n} \Bigr) \mathrm{Var}(U^*_{nj}). 
\end{equation}
For details refer to ~\cite{Mentch2016}.
 The variance of the complete U-statistic $U^*_{nj}$ is given by 
\[
\sum_{c=1}^{k_n} \frac{k_n!^2}{c!\,(k_n-c)!^2}\,
\frac{(n-k_n)(n-k_n-1)\cdots(n-2k_n+c+1)}{n(n-1)\cdots(n-k_n+1)}\,
\zeta^{(j)}_{c,k_n}.
\]
The details of the calculations above are shown in  \cite{VanderVaart2000} page 163.

For the current setup, \[ \zeta^{(j)}_{k_n,k_n} = \mathrm{Var}[h_{nj}(\mathbf{Z}_1,...,\mathbf{Z}_{k_n})] \leq \mathbf{E}\Big[\bar{g}_{nj}^2(\mathbf{Z}_{k_n+1},\mathbf{Z}_1,...,\mathbf{Z}_{k_n}) \Big], \] thus condition $ (C2)$ implies
\[
 \underset{1 \leq j \leq p}{\max}\zeta^{(j)}_{k_n,k_n}  = O(k_n^{\alpha}).
\]
Therefore there exists a constant $M$ such that  $ \underset{1 \leq j \leq p}{\max}\zeta^{(j)}_{k_n,k_n} \leq M k_n^\alpha$, and we can bound the variance of \(U^*_{nj}\) as follows:

\begin{align}
\textcolor{black}{\mathrm{Var}(U^*_{nj})} =
  &\sum_{c=1}^{k_n} \frac{k_n!^2}{c!\,(k_n-c)!^2}\,
    \frac{(n-k_n)(n-k_n-1)\cdots(n-2k_n+c+1)}{n(n-1)\cdots(n-k_n+1)}\,
    \zeta^{(j)}_{c,k_n} \notag \\
    \leq   &\sum_{c=1}^{k_n} \frac{1}{c!}\Big(\frac{k_n^c (k_n -c)!}{(k_n -c)!}\Big)^2 \frac{(n-k_n)^{k_n-c}}{(n-k_n)^{k_n}} \zeta^{(j)}_{c,k_n} \\
  \leq\; &\sum_{c=1}^{k_n}  \frac{k_n^{2c}}{c!}\,
    \frac{1}{(n-k_n)^{c}}\,
    \zeta^{(j)}_{k_n,k_n} \notag \\
  \leq\; &\sum_{c=1}^{k_n} \frac{1}{c!}
    \left(\frac{k_n^{2}}{n-k_n}\right)^{c}  M k_n^{\alpha}\notag \\
  \leq &\Bigl(\tfrac{k_n^{2}}{n-k_n} \exp\!\Bigl(\dfrac{k_n^{2}}{n-k_n}\Bigr)\Bigr) M k_n^\alpha. \notag    
\end{align}

\noindent
Therefore, 
\[
\textcolor{black}{\underset{1 \leq j \leq p}{\max}\,\mathrm{Var}\bigl(U_{nj} \bigr)
\leq\frac{Mk_n^{\alpha}}{m_n} + \Bigl(1 - \frac{1}{m_n}\Bigr) \Bigl(\dfrac{k_n^{2}}{n-k_n} \exp\!\Bigl(\dfrac{k_n^{2}}{n-k_n}\Bigr)\Bigr) M k_n^\alpha
 = O\Big(\frac{k_n^{2+\alpha}}{n}\Big).}
\]
This complete the proof of (1), (2) and (3).

\noindent
Now from parts (1), (2) and (3) we get
\begin{align}\label{gamma j}
     \underset{1 \leq j \leq p}{\max}\mathrm{Var}(\Gamma_{nj}) 
     &\leq O\Big(\frac{n}{m_n}\Big) + O\Big(\frac{k_n^{\alpha }}{n}\Big) + O\Big(\frac{k_n^{2+\alpha}}{n}\Big) = O(R_n)
\end{align}
where $R_n = \max \Big\{ \dfrac{k_n^{2+\alpha}}{n}, \dfrac{n}{m_n} \Big\}$.
Similarly we derive
\begin{equation} \label{gamma}
   \mathrm{Var}(\Gamma_n) = O(R_n).
\end{equation}
The proof follows the same line of argument as above, \textcolor{black}{by changing $\Omega_{mj} = (\omega_m, \pi_{mj})$ to $\omega_m$ for $m=1,..,m_n$.} 
Thus combining results in (\ref{gamma j}) and (\ref{gamma}), we get
\[
\underset{1 \leq j \leq p}{\max}\mathrm{Var}(\hat{\Lambda}_n^{(j)}) = O(R_n).
\]


\noindent
Finally, to prove $\mathrm{P}(\mathcal{D} \subseteq \widehat{\mathcal{D}}_n ) \to 1$, we will prove that $\mathrm{P}(\mathcal{D} \not\subseteq \widehat{\mathcal{D}}_n ) \to 0.$ Note that if $\mathcal{D} \not\subset \widehat{D}_n$, then by
Condition~\ref{threshold} and the selection rule in \eqref{D hat star} there exists some
$d\in \mathcal{D}$ with $\Lambda^{(d)} > c_0$ and $\hat{\Lambda}^{(d)}_{n} \leq c_0/2$ and thus, $|\hat{\Lambda}^{(d)}_{n} - \Lambda^{(d)}| \geq c_0/2$. Therefore,
\begin{align}
\mathrm{P}(\mathcal{D} \not\subset \widehat{\mathcal{D}}_n) \notag
&=  \mathrm{P}\!\left( \bigcup_{d\in \mathcal{D}} \Bigl\{\, d \notin \widehat{\mathcal{D}}_n \,\Bigr\} \right) \notag\\
&\le  \sum_{d \in \mathcal{D}} \,\mathrm{P}\!\left( \Bigl|\,\hat{\Lambda}^{(d)}_{n} - \Lambda^{(d)}\,\Bigr| \geq c_0/2 \right) \notag\\
&\leq \sum_{d \in \mathcal{D}} \frac{4\mathrm{Var}(\hat{\Lambda}_n^{(d)})}{c_0^2} = O\Bigl( \dfrac{\mathrm{Card}(\mathcal{D}) R_n}{c_0^2} \Bigr) \notag.
\end{align}
Thus,
$\mathrm{P}\!\left(\mathcal{D} \subseteq \widehat{\mathcal{D}}_n\right) \to 1$ if $\dfrac{\mathrm{Card}(\mathcal{D}) R_n}{c_0^2}\to 0.$


\end{proof}

\vskip 0.2in

\bibliographystyle{plainnat}
\bibliography{sample}

\end{document}